\documentclass[smallextended]{llncs}

\usepackage{hyperref}
\usepackage{amsthm}
\theoremstyle{plain}
\usepackage{amsfonts}
\usepackage{amssymb}
\usepackage{amsmath}
\usepackage{graphicx}
\usepackage{psfrag}
\usepackage{xcolor}
\usepackage{enumerate}
\usepackage{subcaption}
\usepackage{mathtools}
\usepackage{tikz-cd}
\usepackage[title]{appendix}
\usetikzlibrary{arrows}
\usepackage{booktabs}

\usepackage[all]{xy}

\newtheorem*{theorem*}{Theorem}
\newtheorem*{proposition*}{Proposition}
\newtheorem*{lemma*}{Lemma}

\newcommand{\modified}[1]{\textcolor{black}{#1}}
\newcommand{\inserted}[1]{\textcolor{black}{#1}}

\begin{document}

\title{Steady and ranging sets in graph persistence}

\titlerunning{Steady and ranging}  
%
\author{Mattia G. Bergomi$^1$, Massimo Ferri$^2$, Antonella Tavaglione$^2$}
\authorrunning{M.G. Bergomi, M. Ferri, A. Tavaglione} 
%
%
\institute{$^1$ Independent Researcher, Milan, Italy\\
$^2$ ARCES and Dept. of Mathematics, Univ. of Bologna, Italy
\email{mattiagbergomi@gmail.com, massimo.ferri@unibo.it, antonella.tavaglione@studio.unibo.it},
}

\maketitle              

\begin{abstract}

Topological data analysis can provide insight on the structure of weighted graphs and digraphs. However, some properties underlying a given (di)graph are hardly mappable to simplicial complexes. We introduce \textit{steady} and \textit{ranging} sets: two standardized ways of producing persistence diagrams directly from graph-theoretical features. The two constructions are framed in the context of \textit{indexing-aware persistence functions}. Furthermore, we introduce a sufficient condition for stability. Finally, we apply the steady- and ranging-based persistence constructions to toy examples and real-world applications.

\end{abstract}


\section{Introduction}

Weighted graphs are a common data structure in many real-world scenarios. Recently, persistent homology became a widespread tool for data analysis, classification, comparison, and retrieval. However, this technique is by its very own nature limited to the analysis of weighted simplicial complexes. Although a graph is a one-dimensional complex, relevant information is not always carried by its topology, but, for instance, by graph-theoretical structures. A common choice to overcome this issue is to associate auxiliary simplicial complexes to the graph, see for instance~\cite{bergomi2020topological}. This strategy has been successfully applied in many interesting applications, e.g.~\cite{PeEx*14,LoEx*16,ReNo*17,rieck2018clique,SiGi*18,chowdhury2018persistent,port2018persistent,blevins2020reorderability,vijay2020weighted}.

It is possible to define and compute persistence in other categories than simplicial complexes or topological spaces~\cite{bergomi2019rank,bergomi2021beyond} and, in a different sense, \cite{patel2018generalized,mccleary2020bottleneck,kim2021generalized,mccleary2020edit,govc2021complexes}. We introduce a further class of indexing-aware persistence functions ({\it ip-functions}), defined on $(\mathbb{R}, \le)$-indexed diagrams in a given category, that can be described via persistence diagrams. Additionally, we display a specific way of building ip-functions for filtered graphs and digraphs, introducing the concepts of {\it steady} and {\it ranging} sets.

We are rather far from the categorifications of \cite{BuSc14,Les15,oudot2015persistence,de2018theory,govc2021complexes}: we aim to provide a simple and agile tool that can be applied directly to graphs (i.e., without mapping graphs to simplicial complexes), and possibly to other structures arising naturally from applications. \modified{The constructions derived from the framework we propose have a topological counterpart obtainable considering the simplicial complex associated with a poset (see \cite[Rem. 1]{bergomi2021beyond}). Here, we show how to bypass that topological construction.}

\modified{Section~\ref{sec:perdia} briefly recalls the classical notions of persistence diagram and bottleneck distance. Section~\ref{sec:grape} focuses on graphs. First, we define {\it ip-functions}, and {\it balanced} ip-functions and discuss their stability. Then, we introduce {\it steady} and {\it ranging sets} as swift generators of ip-functions based directly on graph-theoretical features. These constructions are the theoretical core of the work. Thereafter, we apply them to study \textit{persistent} Eulerian sets and monotone features on some elementary graphs. Section~\ref{sec:hubs} showcases how the steady and ranging constructions can be leveraged in hub-detection tasks. Concrete applications follow in Section~\ref{sec:concrete}: we compute steady and ranging hubs in a network of airports, the character co-occurrence networks of {\it Les Mis\'erables} and {\it Game of Thrones}, and a set of languages. Section~\ref{sec:digraph} extends to weighted digraphs the theory developed in the previous sections. Code for application is available as a Python package at the repository \href{https://github.com/MGBergomi/hubpersistence.git}{https://github.com/MGBergomi/hubpersistence.git}.} The Appendix contains examples showing that most ip-functions of the paper are not balanced.

\subsection{Persistence diagrams} \label{sec:perdia}

The main object of study in persistent homology \cite{EdHa08} are filtered spaces, i.e. pairs $(X, f)$ where $X$ is a topological space (e.g., the space of a simplicial complex) and $f:X \to \mathbb{R}$ is a map called {\it filtering function}: sublevel sets $X_u= f^{-1}\big((-\infty, u]\big)$ are compared through homology morphisms induced by inclusion, in particular the so-called Persistent Betti Number functions. From such a function a {\it persistence diagram} (see Def.~\ref{def:perdia}) can be built \cite[Sect. 2]{CoEdHa07}. In turn, Persistent Betti Number functions can be recovered from the persistence diagram, \cite{CoEdHa07}.

Persistence diagrams are the most widely used ``fingerprints'' of filtered spaces. The {\it bottleneck distance} between persistence diagrams yields an effective lower bound to distances between filtered spaces. This makes persistence diagrams a powerful tool in shape classification, analysis and retrieval.
The strategic advantage of the generalisation started in \cite{bergomi2019rank,bergomi2021beyond} consists in the fact that also categorical persistence functions (Def.~\ref{def:persistence}) can be represented by persistence diagrams: see \cite[Sec. 3.9]{bergomi2019rank}.

In $\mathbb{R}\times(\mathbb{R}\cup\{+\infty\})$ set $\Delta=\{(u, v) \, | \, u=v\}$, $\Delta^+=\{(u,v) \, | \, u<v \}$ and $\bar{\Delta}^+ = \Delta \cup \Delta^+$. In a multiset, the {\it multiplicity} of an element will be the number of times that the element appears.

\begin{definition}\cite{CoEdHa07,ChaCo*09} \label{def:perdia}
A \emph{persistence diagram} $D$ is a multiset of points of $\bar{\Delta}^+$ where every point of the \emph{diagonal} $\Delta$ appears with infinite multiplicity.
\end{definition}

The points of $D$ belonging to $\Delta^+$ are called {\it cornerpoints}; they are said to be {\it proper} if both their coordinates are finite, {\it cornerpoints at infinity} otherwise. A persistence diagram is said to be {\it finite} if so is its set of cornerpoints. We shall only consider finite persistence diagrams.

\begin{definition}\label{def:bottleneck}
Given persistence diagrams $D, D'$, let $\Gamma$ be the set of all bijections between $D$ and $D'$. We define the \emph{bottleneck} (formerly \emph{matching}) \emph{distance} as the real number $$d(D, D') = \inf_{\gamma\in\Gamma} \sup_{p\in D} \| p-\gamma(p)\|_\infty$$
\end{definition}

First, this distance function checks the maximum displacement between corresponding points for a given matching either between cornerpoints of the two diagrams or cornerpoints and their projections on the diagonal $\Delta$. Then, the minimum among these maxima is computed. Minima and maxima are actually attained because of the requested finiteness.

\section{Graph-theoretical persistence}\label{sec:grape}

Let {\bf Graph} be the category having finite simple undirected graphs as objects and injective simplicial applications as morphisms, seen as a subcategory of the category of finite simplicial complexes. In what follows, a graph will be considered as the pair of its vertex set and edge set, i.e. $G=(V, E)$, $G'=(V', E')$ and so on.

\modified{
\begin{definition}\cite[Sect. 1.3]{BuSc14}\label{def:indexed}
   An \emph{$(\mathbb{R}, \le)$-indexed diagram} is any functor from the category $(\mathbb{R}, \le)$ to an arbitrary category $\mathbf{C}$. $(\mathbb{R}, \le)$-indexed diagrams form a category, $\mathbf{C}^{(\mathbb{R}, \le)}$. The $(\mathbb{R}, \le)$-indexed diagram is said to be \emph{monic} if all morphisms of its image are monomorphisms of $\mathbf{C}$.
\end{definition}
}

We consider $(\mathbb{R}, \le)$-indexed diagrams in {\bf Graph} that are constant on a finite set of left-closed, right-open intervals. Because of the choice of monomorphisms as the only acceptable morphisms, every such $(\mathbb{R}, \le)$-indexed diagram is monic, \modified{see Def.~\ref{def:indexed},}  and can be seen, up to natural isomorphisms, as a filtration of a graph $G$ coming from a filtering function $f:V\cup E \to \mathbb{R}\cup \{+\infty\}$. Moreover, we shall limit our study to $(\mathbb{R}, \le)$-indexed diagrams whose associated filtration has no isolated vertices at any level. In other words, the filtering function $f$ takes value $+\infty$ if a vertex is isolated, and the minimum of its values on the edges incident to the vertex, otherwise. Thus, $f$ is determined by its restriction to $E$; therefore the {\it weighted graphs} considered here are pairs $(G, f)$ with $f:E \to \mathbb{R}$\modified{. By construction, the subgraphs of the corresponding filtrations are induced by their edge sets}.

\begin{figure}[htb]
\centering
\includegraphics[width =0.75 \textwidth]{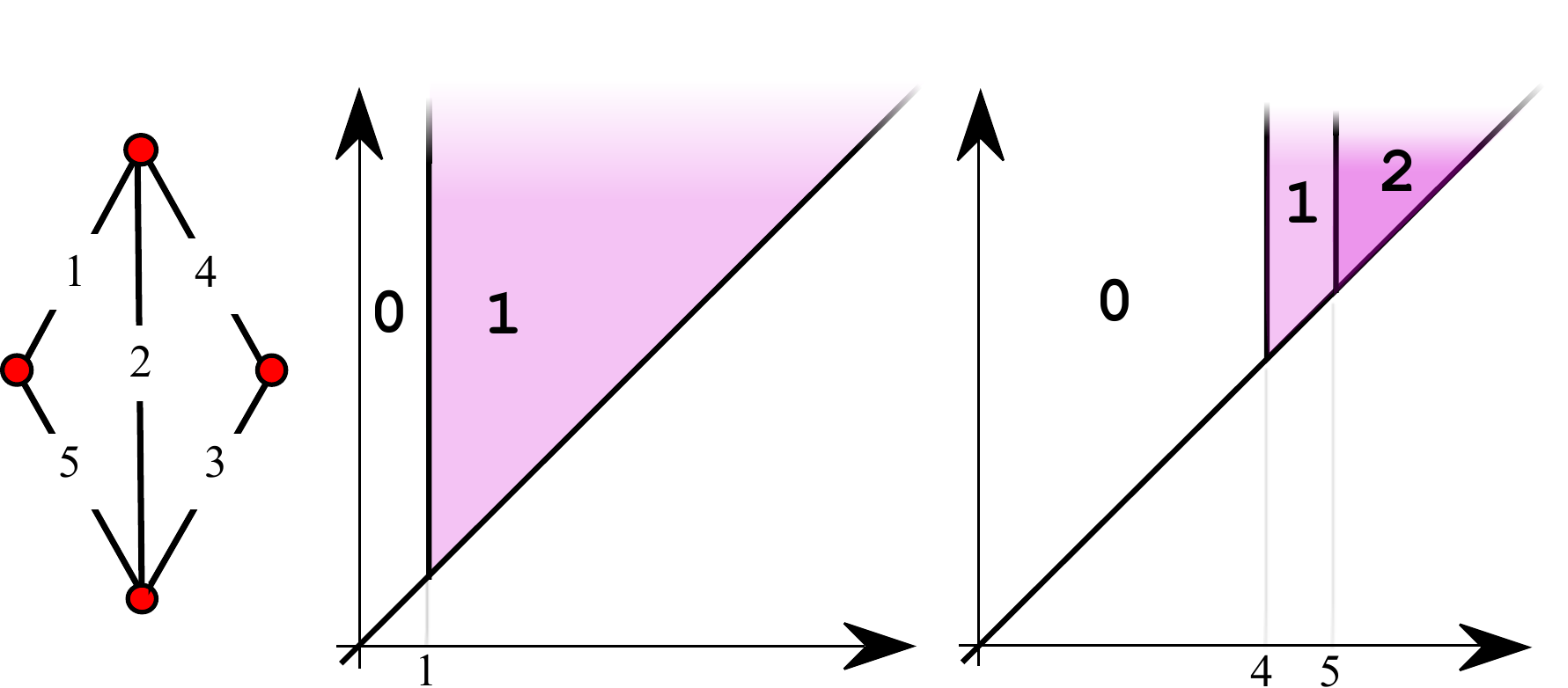}
\caption{A weighted graph (left) and its Persistent Betti Number functions in degree 0 (middle) and 1 (right).
}\label{fig:theta}
\end{figure}

\modified{
\begin{definition}\cite[Def. 3.2]{bergomi2019rank}
    \label{def:persistence}
    Let $\mathbf{\bar{C}}$ be a category. A lower-bounded function $p:\textnormal{Morph}(\mathbf{\bar{C}}) \rightarrow\mathbb{Z}$ is a \emph{categorical persistence function} if, for all $u_1 \rightarrow u_2 \rightarrow v_1 \rightarrow v_2$, the following inequalities hold:
\begin{enumerate}
    \item $p(u_1\rightarrow v_1) \le p(u_2\rightarrow v_1)$ and $p(u_2\rightarrow v_2) \le p(u_2\rightarrow v_1)$.
    \item $p(u_2\rightarrow v_1) - p(u_1\rightarrow v_1) \ge p(u_2\rightarrow v_2) - p(u_1\rightarrow v_2)$.
\end{enumerate}
\end{definition}
}

\begin{remark}\label{rem:vect}
Such a function is categorical in the sense that it yields the same result to morphisms obtained from each other by composition with a $\mathbf{\bar{C}}$-isomorphism.
For instance, we can retrieve the framework of classical topological persistence by setting $\mathbf{\bar{C}}= \mathbf{Vect}$ and $p$ as the rank operator, i.e. the dimension of the image.
\end{remark}

In what follows we focus on $\mathbf{\bar{C}}=(\mathbb{R}, \le)$. In this case a morphism $u\rightarrow v$ is simply the relation $u\le v$, which is represented as the point $(u,v)$ in the persistence diagrams.

\modified{
\begin{definition}\label{def:genpersfct}
Let $p$ be a map assigning to each monic $(\mathbb{R}, \le)$-indexed diagram $M$ in a category $\mathbf{C}$ a categorical persistence function $p_M$ on $(\mathbb{R}, \le)$, such that $p_{M} = p_{M'}$ whenever a natural isomorphism between $M$ and $M'$ exists. All the resulting categorical persistence functions $p_M$ are called \emph{indexing-aware persistence functions in} $\mathbf{C}$ (\emph{ip-functions} for brevity). The map $p$ itself is called an \emph{ip-function generator}.
\end{definition}
}

\begin{remark}\label{rem:generators}
An ip-function generator is actually a categorical function (in the sense of Rem.~\ref{rem:vect}) on the functor category $\mathbf{C}^{(\mathbb{R}, \le)}$ .
\end{remark}

\modified{An ip-function in {\bf Graph} (Def.~\ref{def:genpersfct}) $p_M$, where $M$ is an  $(\mathbb{R}, \le)$-indexed diagram, will be denoted $p_{(G, f)}$, where $M$ corresponds to the filtration produced by the weighted graph $(G, f)$. The associated persistence diagram will be denoted by $D(f)$, for the sake of simplicity and if no confusion may occur.}

\modified{We can now observe that ip-functions are a particular case of categorical persistence functions in the category {\bf Graph}. We recall that categorical persistence functions generalise Persistent Betti Number (PBN) functions. The difference between any of the categorical persistence functions introduced in \cite{bergomi2021beyond} and an ip-function defined here is that the former comes from a functor defined on {\bf Graph}, while the latter strictly depends on the filtration, so comes from a functor defined on $(\mathbb{R}, \le)$.}

\modified{
    \begin{remark}\label{rem:running_ex}
        The graph depicted in Fig.~\ref{fig:theta} shall be our running toy example along the entire manuscript. In the figure, we report the PBN functions of degree $0$ and $1$ to allow the reader to compare those classical results with the ones we shall obtain through ip-functions.
    \end{remark}
}

\modified{
In Section~\ref{sec:digraph}, we extend the notions introduced above to the category of directed graphs.}

\subsection{Balanced ip-functions}\label{sec:balfct}

The categorical functions introduced in \cite{bergomi2021beyond} are stable, i.e. the bottleneck distance between their persistence diagrams is a lower bound for their interleaving distance. The same does not automatically hold for ip-functions.
However, we shall state a condition (Def.~\ref{def:balfct}) which implies stability (as proved in Thm.~\ref{thm:stab}). This condition corresponds to \cite[Prop. 10]{dAFrLa10}: there, it is proved for 0-degree PBNs, and from it the stability theorem \cite[Thm. 29]{dAFrLa10} follows through a sequence of lemmas; here, it is postulated.

\begin{definition}\label{def:balfct}
Let $p$ be a ip-function generator on {\bf Graph}. The map $p$ itself and the resulting ip-functions are said to be \emph{balanced} if the following condition is satisfied. Let $(G, f)$ and $(G', f')$ be any two weighted graphs, and $p_{(G, f)}$, $p_{(G', f')}$ their associated ip-functions. If an isomorphism $\psi:G\to G'$ \modified{and a positive real number $h$ exists, such that $\sup_{e\in E} |f(e)-f'\big(\psi(e)\big)|\le h$}, then for all $(u, v)\in \Delta^+$ the inequality $p_{(G, f)}(u-h, v+h)\le p_{(G', f')}(u, v)$ holds.
\end{definition}

Let $(G, f)$, $(G', f')$ be as above. Let also $\mathcal{H}$ be the (possibly empty) set of graph isomorphisms between $G$ and $G'$. We can now take to {\bf Graph} some definitions given in \cite{FrMu99,dAFrLa10,Les15}.

\begin{definition}\label{def:natural}
The \emph{natural pseudodistance} of $(G, f)$ and $(G', f')$ is
$$\delta\big((G, f), (G', f')\big) = \left\{
\begin{array}{l l}
+\infty  & \textnormal{if} \ \ \ \mathcal{H}=\emptyset \\
\inf_{\phi\in \mathcal{H}}\sup_{e\in E} |f(e) - g\big(\phi(e)\big)| \ \ & \textnormal{otherwise}
\end{array}
\right.$$
\end{definition}

Some simple adjustments of the proof of \cite[Thm. 29]{dAFrLa10} and of its preceding lemmas yield the following theorem.

\begin{theorem}[Stability]\label{thm:stab}
Let $p$ be a balanced ip-function generator in {\bf Graph} and $(G, f), (G', f')$ be two weighted graphs. Then we have
\[d\big(D(f), D(f')\big) \le \delta\big((G, f), (G', f')\big),\]
where $D(f)$ and $D(f')$ are the persistence diagrams realized by the ip-functions $p_{(G, f)}$ and $p_{(G', f')}$ respectively.
$\square$
\end{theorem}

Through \cite[Thm. 5.8]{FrLaMe19}, this also implies stability with respect to the interleaving distance.  
Universality \cite[Sec. 5.2]{Les15} is generally not granted for stable persistence functions: it needs {\it ad hoc} constructions.

\modified{
    When discussing stability above, we introduced two distinct graphs. However, the following proposition describes stability when considering a single graph and two filtering functions. This result will be useful in the remainder of the paper.
\begin{proposition}\label{same}
The ip-function generator $p$ is balanced if and only if the following condition is satisfied.
Let $G=(V, E)$ be any graph, $f$ and $g$ be two filtering functions on $G$, and $p_{G, f)}$ and $p_{(G, g)}$ their ip-functions. If a positive real number $h$ exists, such that $\sup_{e\in E}|f(e)-g(e)|\le h$, then for all $(u, v)\in \Delta^+$ the inequality $p_{(G,f)}(u-h, v+h) \le p_{(G, g)}(u, v)$ holds.
\end{proposition}
\begin{proof}
One of the two implications is immediate. The other is proved by the fact that $p_{(G, g)} = p_{(G', f')}$ where $g = f'\circ \psi$, with the notation of Def.~\ref{def:balfct}.
\end{proof}
\begin{remark}
The condition is symmetric: if it holds as in the statement of Prop.~\ref{same}, then also  $p_{(G, g)}(u-h, v+h) \le p_{(G, f)}(u, v)$ holds for all $(u, v)\in \Delta^+$.
\end{remark}
}

\subsection{Steady and ranging sets}\label{sec:sr}

\begin{definition}\label{def:feature}
Given a graph $G = (V, E)$, any function $\mathcal{F}:2^{V\cup E} \to \{true, false\}$ is called a \emph{feature}. We call $\mathcal{F}$-set any $X\subset V\cup E$ such that $\mathcal{F}(X)= true$. Given a weighted graph $(G, f)$ and a  real number $u$, we denote by $G_u$ the subgraph of $G$ induced by the edge set $f^{-1}(-\infty, u]$. We shall say that $X\subset V\cup E$ is an \emph{$\mathcal{F}$-set at level} $w\in \mathbb{R}$ if it is an $\mathcal{F}$-set of the subgraph $G_w$.
\end{definition}

\begin{definition}\label{maximal}
Let $\mathcal{F}$ be a feature of $G$. We define the {\em maximal} feature $m\mathcal{F}$ associated with $\mathcal{F}$ as follows: for any $X \subseteq (V \cup E)$, $m\mathcal{F}(X) = true$ if and only if $\mathcal{F}(X)=true$ and there is no $Y \subseteq (V \cup E)$ such that $X \subset Y$ and $\mathcal{F}(Y)=true$.
\end{definition}

\begin{definition}\label{def:sr}
Let $\mathcal{F}$ be a feature. A set $X\subseteq V\cup E$ is a \emph{steady} $ \mathcal{F}$-set (s$\mathcal{F}$-set for brevity) at $(u, v) \in \Delta^+$ if it is an $\mathcal{F}$-set at all levels $w$ with $u\le w \le v$. We call $X$ a \emph{ranging} $\mathcal{F}$-set (r$\mathcal{F}$-set) at $(u, v)$ if there exist levels $w\le u$ and $w'\ge v$ at which it is an $\mathcal{F}$-set.

Let $S^\mathcal{F}_{(G, f)}(u,v)$ be the set of s$\mathcal{F}$-sets at $(u, v)$ and let $R^\mathcal{F}_{(G, f)}(u,v)$ be the set of r$\mathcal{F}$-sets at $(u, v)$.
\end{definition}

\begin{remark}\label{rem:implies} \modified{Intuitively, the adjective ``steady'' stresses that a steady set enjoys a given feature $\mathcal{F}$ throughout the entire interval $[u, v)$.  ``Ranging'', instead, refers to the fact that a ranging set spans, with feature $\mathcal{F}$, the range $[u, v)$ although possibly with gaps.} 
Of course, steady implies ranging. This implication is granted by the ``$\le$'' and ``$\ge$'' signs in the definitions. With strict inequalities the implication fails. There are features for which steady is equivalent to ranging, e.g., features for which a set can be an $\mathcal{F}$-set only in a (possibly unbounded) interval. A simple example is the feature $\mathcal{F}$ which assigns $true$ only to singletons consisting of a vertex of a fixed degree. 
\end{remark}

\begin{lemma}\label{superlemma}
If $u\le u' < v' \le v$, then
\begin{enumerate}
\item $S^\mathcal{F}_{(G, f)}(u,v) \subseteq S^\mathcal{F}_{(G, f)}(u',v')$
\item $R^\mathcal{F}_{(G, f)}(u,v) \subseteq R^\mathcal{F}_{(G, f)}(u',v')$
\end{enumerate}
where the equalities hold if $G_u = G_{u'}$ and $G_v = G_{v'}$. Moreover $S^\mathcal{F}_{(G, f)}(u,v) = \emptyset = R^\mathcal{F}_{(G, f)}(u, v)$ if $G_u =\emptyset$.
\end{lemma}
\begin{proof}
By the definitions themselves of steady and ranging $\mathcal{F}$-set.
\end{proof}

\begin{definition}\label{def:rhosigma}
Let $\mathcal{F}$ be a feature. For any graph $G$, for any filtering function $f:E \to \mathbb{R}$, we define $\sigma^\mathcal{F}_{(G, f)}: \Delta^+ \to \mathbb{Z}$ as  the function which assigns to $(u, v) \in \Delta^+$ the number $|S^\mathcal{F}_{(G, f)}(u,v)|$ and $\varrho^\mathcal{F}_{(G, f)}: \Delta^+ \to \mathbb{Z}$ as  the function which assigns to $(u, v) \in \Delta^+$ the number $|R^\mathcal{F}_{(G, f)}(u,v)|$. We denote by  $\sigma^\mathcal{F}$ and $\varrho^\mathcal{F}$ the maps assigning $\sigma^\mathcal{F}_{(G, f)}$ and $\varrho^\mathcal{F}_{(G, f)}$ respectively to the $(\mathbb{R}, \le)$-indexed diagram corresponding to $(G, f)$.
\end{definition}

\begin{proposition}\label{supergp}
The maps $\sigma^\mathcal{F}$ and $\varrho^\mathcal{F}$ are ip-function generators.
\end{proposition}
\begin{proof}
We prove conditions 1 and 2 of Def.~\ref{def:persistence}, recalling that the source category is $(\mathbb{R}, \le)$, so the existence of a morphism $u \to v$ (with $u\ne v$) simply means that $u<v$. Assume $u_1<u_2<v_1<v_2$. Let $(G, f)$ be any weighted graph.
\begin{itemize}
\item{\bf (Condition 1 for $\sigma^\mathcal{F}$)}
By Lemma~\ref{superlemma}, $S^\mathcal{F}_{(G, f)}(u_1, v_1) \subseteq S^\mathcal{F}_{(G, f)}(u_2, v_1)$, so $|S^\mathcal{F}_{(G, f)}(u_1, v_1)| \le |S^\mathcal{F}_{(G, f)}(u_2, v_1)|$. Also $S^\mathcal{F}_{(G, f)}(u_2, v_2) \subseteq S^\mathcal{F}_{(G, f)}(u_2, v_1)$ and $|S^\mathcal{F}_{(G, f)}(u_2, v_2)| \le |S^\mathcal{F}_{(G, f)}(u_2, v_1)|$.
\item{\bf (Condition 2 for $\sigma^\mathcal{F}$)}
By Lemma~\ref{superlemma}, $S^\mathcal{F}_{(G, f)}(u_1, v_1) \subseteq S^\mathcal{F}_{(G, f)}(u_2, v_1)$,  so $|S^\mathcal{F}_{(G, f)}(u_2, v_1)| - |S^\mathcal{F}_{(G, f)}(u_1, v_1)|$ is the number of s$\mathcal{F}$-sets at $(u_2, v_1)$ which fail to be $\mathcal{F}$-sets at some $w$ with $u_1\le w \le u_2$. Analogously for $|S^\mathcal{F}_{(G, f)}(u_2, v_2)| - |S^\mathcal{F}_{(G, f)}(u_1, v_2)|$.
\newline
Now, every s$\mathcal{F}$-set at $(u_1, v_2)$ which fails to be an $\mathcal{F}$-set at $w$ with $u_1\le w \le u_2$ is also an s$\mathcal{F}$-set at $(u_1, v_1)$ failing at the same $w$. So $S^\mathcal{F}_{(G, f)}(u_2, v_1) - S^\mathcal{F}_{(G, f)}(u_1, v_1) \supseteq S^\mathcal{F}_{(G, f)}(u_2, v_2) - S^\mathcal{F}_{(G, f)}(u_1, v_2)$ and $|S^\mathcal{F}_{(G, f)}(u_2, v_1)| - |S^\mathcal{F}_{(G, f)}(u_1, v_1)| \geq |S^\mathcal{F}_{(G, f)}(u_2, v_2)| - |S^\mathcal{F}_{(G, f)}(u_1, v_2)|$.
\item{\bf (Condition 1 for $\varrho^\mathcal{F}$)}
The argument is the same as for $\sigma^\mathcal{F}$.
\item{\bf (Condition 2 for $\varrho^\mathcal{F}$)}
By Lemma~\ref{superlemma}, $R^\mathcal{F}_{(G, f)}(u_1, v_1) \subseteq R^\mathcal{F}_{(G, f)}(u_2, v_1)$,  so $|R^\mathcal{F}_{(G, f)}(u_2, v_1)| - |R^\mathcal{F}_{(G, f)}(u_1, v_1)|$ is the number of r$\mathcal{F}$-sets at $(u_2, v_1)$ which fail to be $\mathcal{F}$-sets at all levels $w$ with $w \le u_1$. Analogously for $|R^\mathcal{F}_{(G, f)}(u_2, v_2)| - |R^\mathcal{F}_{(G, f)}(u_1, v_2)|$.
\newline
Now, every r$\mathcal{F}$-set at $(u_1, v_2)$ which fails to be an $\mathcal{F}$-set at all levels $w$ with $w \le u_1$ is also an r$\mathcal{F}$-set at $(u_1, v_1)$ failing at the same levels $w$. So $R^\mathcal{F}_{(G, f)}(u_2, v_1) - R^\mathcal{F}_{(G, f)}(u_1, v_1) \supseteq R^\mathcal{F}_{(G, f)}(u_2, v_2) - R^\mathcal{F}_{(G, f)}(u_1, v_2)$ and $|R^\mathcal{F}_{(G, f)}(u_2, v_1)| - |R^\mathcal{F}_{(G, f)}(u_1, v_1)| \geq |R^\mathcal{F}_{(G, f)}(u_2, v_2)| - |R^\mathcal{F}_{(G, f)}(u_1, v_2)|$.
\end{itemize}
\end{proof}

The value of both functions $\sigma^\mathcal{F}_{(G, f)}$ and $\varrho^\mathcal{F}_{(G, f)}$ at a point $P$ on a vertical (resp. horizontal) discontinuity line is the same as the value at the points in a right (resp. upper) neighbourhood of $P$

Of course, there are many features which give valid but meaningless ip-functions: the features $\mathcal{F}$ such that, if $X$ is an $\mathcal{F}$-set at level $u$, then it is an $\mathcal{F}$-set also at level $v$ for all $v>u$.

We still don't know which general hypothesis on $\mathcal{F}$ would imply that $\sigma^\mathcal{F}$ or $\varrho^\mathcal{F}$ are balanced ip-function generators (Def.~\ref{def:balfct}). Such features exist: \modified{Section~\ref{sec:monotone} presents a whole class of features giving rise to balanced ip-functions}.

\subsection{Steady and ranging persistence on Eulerian sets}\label{sec:euler}

We now give an example of the framework exposed in Section~\ref{sec:sr}. Given any graph $G$, we define $\mathcal{EU}: 2^{V\cup E} \to \{true, false\}$ to yield $true$ on a set $A$ if and only if $A$ is a set of vertices whose induced subgraph of $G$ is nonempty, Eulerian and maximal with respect to these properties; in that case $A$ is said to be a $\mathcal{EU}$-\emph{set} of $G$. \modified{$\mathcal{EU}$ is then the maximal version of a feature we are not going to deal with.} Let now $(G, f)$ be a weighted graph. We apply Def.~\ref{def:sr} to feature $\mathcal{EU}$.

\begin{definition}
For any real number $w$, the subset $A\subseteq V$ is a \emph{$\mathcal{EU}$-set at level} $w$ if it is a $\mathcal{EU}$-set of the subgraph $G_w$. It is a \emph{steady} $\mathcal{EU}$-set (an s$\mathcal{EU}$-set) at $(u, v) \in \Delta^+$ if it is a $\mathcal{EU}$-set at all levels $w$ with $u\le w \le v$. It is a \emph{ranging} $\mathcal{EU}$-set (an r$\mathcal{EU}$-set)  at $(u, v)$ if there exist levels $w\le u$ and $w'\ge v$ at which it is a $\mathcal{EU}$-set.
\newline
$S^{\mathcal{EU}}_{(G, f)}(u, v)$ and $R^{\mathcal{EU}}_{(G, f)}(u, v)$ are respectively the sets of s$\mathcal{EU}$-sets and of r$\mathcal{EU}$-sets at $(u, v)$. We define $\sigma^\mathcal{EU}_{(G, f)}: \Delta^+ \to \mathbb{R}$ as  the function which assigns to $(u, v) \in \Delta^+$ the number $|S^\mathcal{EU}_{(G, f)}(u,v)|$ and $\varrho^\mathcal{EU}_{(G, f)}: \Delta^+ \to \mathbb{R}$ as  the function which assigns to $(u, v) \in \Delta^+$ the number $|R^\mathcal{EU}_{(G, f)}(u,v)|$.
\newline
We denote by  $\sigma^\mathcal{EU}$ and $\varrho^\mathcal{EU}$ the maps assigning $\sigma^\mathcal{EU}_{(G, f)}$ and $\varrho^\mathcal{EU}_{(G, f)}$ respectively to the $(\mathbb{R}, \le)$-indexed diagram corresponding to $(G, f)$. By Prop.~\ref{supergp}, $\sigma^\mathcal{EU}$ and $\varrho^\mathcal{EU}$ are ip-function generators.
\end{definition}

\begin{figure}[htb]
\centering
\includegraphics[width =0.75 \textwidth]{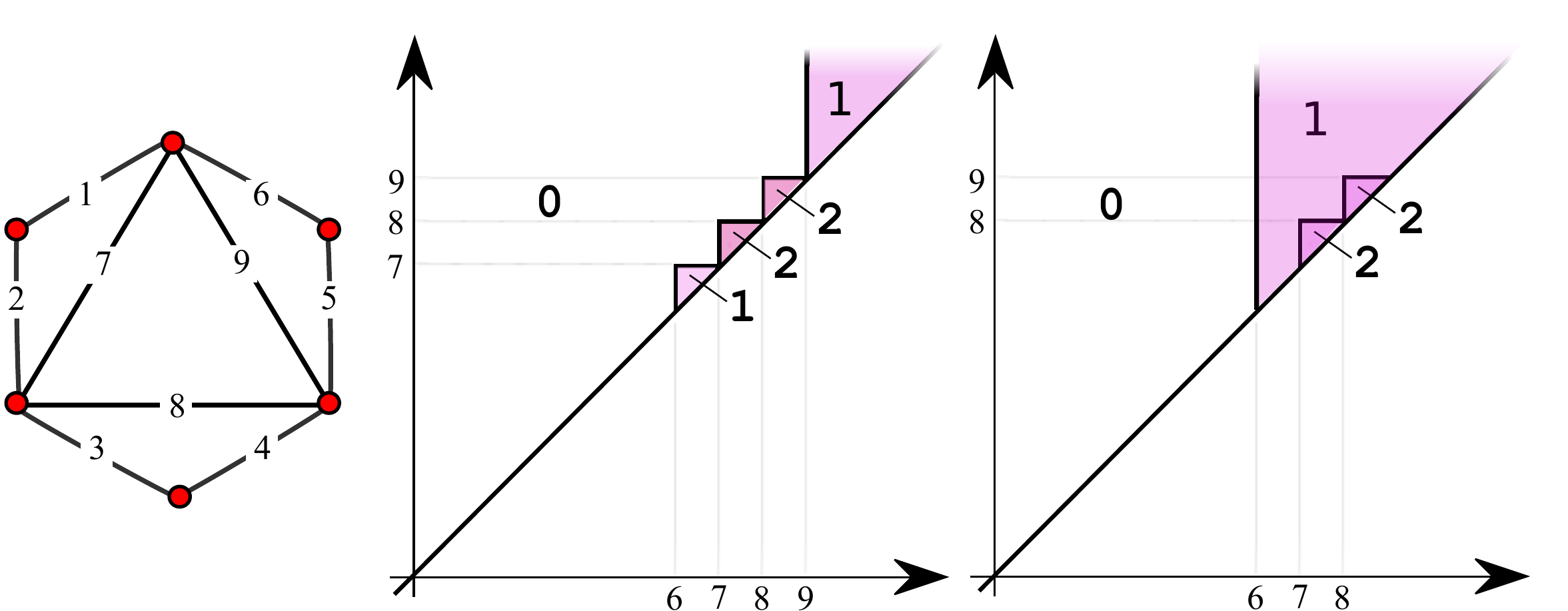}
\caption{A weighted graph $(H, h)$ (left) and the corresponding functions $\sigma^\mathcal{EU}_{(H, h)}$ (middle) and $\varrho^\mathcal{EU}_{(H, h)}$ (right).
}\label{fig:exa}
\end{figure}

\modified{Consider the example displayed in Fig.~\ref{fig:theta}. In that particular example, the functions $\sigma^\mathcal{EU}_{(G, f)}$ and $\varrho^\mathcal{EU}_{(G, f)}$ are the same. Furthermore, they also coincide with the PBN function in degree $1$ shown in the same figure. We show that this is not always the case in Fig.~\ref{fig:exa}.}

Both functions $\sigma^\mathcal{EU}$ and $\varrho^\mathcal{EU}$ are not balanced (see the Appendix).

\inserted{
\subsection{Monotone features}\label{sec:monotone}
For a given graph $G = (V, E)$, we shall consider as subgraphs only the ones induced by sets of edges. The next definition is a variation on the notion of monotone (sometimes dubbed hereditary) property defined in \cite{alon2008every}.
}
\inserted{
\begin{definition}\label{def:monotone}
We say that a feature $\mathcal{F}$ is {\em monotone} if
\begin{itemize}
\item for any graphs $G' =(V', E')\subset G''=(V'', E'')$, and any $X \subseteq (V' \cup E')$, $\mathcal{F}(X) = true$ in $G''$ implies $\mathcal{F}(X)= true$ in $G'$
\item in any graph $\overline{G}=(\overline{V}, \overline{E})$, for any $Y \subset X \subseteq \overline{V} \cup \overline{E}$, $\mathcal{F}(X) = true$ implies $\mathcal{F}(Y)= true$.
\end{itemize}
\end{definition}
}

\inserted{
A paradigmatic monotone feature is independence: independent (or stable) sets and matchings are examples of sets of vertices, respectively of edges, with monotone features.
}

\inserted{
For the remainder of this section, let $(G, f)$ be a weighted graph, $G=(V, E)$, and $\mathcal{F}$ a monotone feature in $G$. By Prop.~\ref{supergp},  $\sigma^\mathcal{F}$ and $\varrho^\mathcal{F}$ are ip-function generators.
}

\inserted{
\begin{lemma}\label{interval}
Let $X \subseteq (V \cup E)$. Then, either there is no value $u$ for which $\mathcal{F}(X)=true$ in $G_u$, or $\mathcal{F}(X)=true$ in $G_u$ for all $u \in [u_1, v_1)$, where $u_1$ is the lowest value $u$ such that in the subgraph $G_u=(V_u, E_u)$ one has $X \subseteq (V_u \cup E_u)$, and $v_1$ is either the lowest value $v$ for which $\mathcal{F}(X)=false$ in $G_v$ or $+\infty$.
\end{lemma}
\begin{proof}
Assume that $\mathcal{F}(X)=true$ in $G_u$ for at least one value $u$.
If $\mathcal{F}(X)=true$ in $G_u$, then $\mathcal{F}(X)=true$ in $G_{u'}=(V_{u'}, E_{u'})$ for all $u'<u$ such that $X\subseteq (V_{u'} \cup E_{u'})$ by Def.~\ref{def:monotone}.
\end{proof}
The interval $[u_1, v_1)$ of Lemma~\ref{interval}, i.e. the widest interval for which $\mathcal{F}(X)=true$ in $(G, f)$, is called the {\it $\mathcal{F}$-interval of $X$ in $(G,f)$}.
\begin{proposition}\label{sigmaro}
$\sigma^\mathcal{F} =\varrho^\mathcal{F}$
\end{proposition}
\begin{proof}
By Lemma~\ref{interval}.
\end{proof}
Let now $g$ be another filtering function on $G$; in order to avoid confusion, for each real number $u$, we denote by $G_{f, u}$ (resp. $G_{g, u}$) the subgraph of $G$ induced by the edge set $f^{-1}\big((-\infty, u]\big)$ (resp. $g^{-1}\big((-\infty, u]\big)$).
\begin{lemma}\label{twofcts}
Assume that there exists a positive real $h$ such that $\sup_{e\in E}|f(e)-g(e)|\le h$. Assume also that $X \subseteq (V \cup E)$ exists, such that $u \in [u_1, v_1)$ is its $\mathcal{F}$-interval in $G, f)$, with $u_1+2h<v_1<+\infty$. Then there is a non-empty $\mathcal{F}$-interval $[u_2, v_2)$ of $X$ in $(G, g)$, and $|u_1-u_2|\le h, |v_1-v_2|\le h$.
\end{lemma}
\begin{proof}
Assume that, for $e\in E$, $f(e)=u$; then $g(e)\le u+h$. This proves that, for each $u$, $G_{f, u}$ is a subgraph of $G_{g, u+h}$. Swapping the roles, also $G_{g, u}$ is a subgraph of $G_{f, u+h}$.
\newline
Therefore, if $X$ exists in $G_{f, u}$ it also exists in $G_{g, u+h}$.Symmetrically, if $X$ exists in $G_{g, u}$ it also exists in $G_{f, u+h}$. Recalling, by Lemma~\ref{interval}, the meaning of $u_1$ and, correspondingly, $u_2$, we obtain that $|u_1-u_2|\le h$.
\newline
If $\mathcal{F}(X)=true$ in $G_{f, u+h}$, then $\mathcal{F}(X)=true$ also in the subgraph $G_{g, u}$ because $\mathcal{F}$ is monotone. Analogously, $\mathcal{F}(X)=true$ in $G_{g, u+h}$ implies $\mathcal{F}(X)=true$ in $G_{f, u}$. Recalling, by Lemma~\ref{interval}, the meaning of $v_1$ and, correspondingly, of $v_2$, we obtain that $|v_1-v_2|\le h$.
\end{proof}
\begin{proposition}\label{monbal}
The ip-function generators $\sigma^\mathcal{F} =\varrho^\mathcal{F}$ are balanced.
\end{proposition}
\begin{proof}
We shall prove for $\sigma^\mathcal{F}$ (and consequently for $\varrho^\mathcal{F}$, by Prop.~\ref{sigmaro}) the property stated in Prop.~\ref{same}. With the notation and the assumptions of Lemma~\ref{interval}, assume that for $u<v$ we have $\sigma^\mathcal{F}_{(G, f)}(u-h, v+h)>0$ (if it vanishes the claim is trivially true). We want to show that $\sigma_{(G, f)}^\mathcal{F}(u-h, v+h)\le \sigma^\mathcal{F}_{(G, g)}(u, v)$.
\newline
Let $X \subseteq (V \cup E)$ be such that $\mathcal{F}(X)= true$ in $G_{f, w}$ for all $w \in [u-h, v+h]$. Then, for the $\mathcal{F}$-interval $[u_1, v_1)$ of $X$ in $(G, f)$ we have $u_1\le u-h$, $v+h<v_1$. The $\mathcal{F}$-interval of the same $X$ in $(G, g)$ is $[u_2, v_2)$, with $|u_1-u_2|\le h$, $|v_1-v_2|\le h$ by Lemma~\ref{twofcts}. So, $u_2\le u_1+h \le u-h+h = u$ and $v= v+h-h < v_1-h \ge v_2$, i.e. $[u, v]$ is contained in the $\mathcal{F}$-interval of $X$ in $(G, g)$ and $\mathcal{F}(X)= true$ in $G_{g, w}$ for all $w\in [u,v]$. Therefore, an injective map exists from $S^\mathcal{F}_{(G,f)}(u-h, v+h)$ to $S^\mathcal{F}_{(G, g)}(u, v)$, proving that $\sigma^\mathcal{F}_{(G, f)}(u-h, v+h) \le \sigma^\mathcal{F}_{(G, g)}(u, v)$.
\end{proof}
}

\begin{figure}[htb]
\centering
\includegraphics[width =0.75 \textwidth]{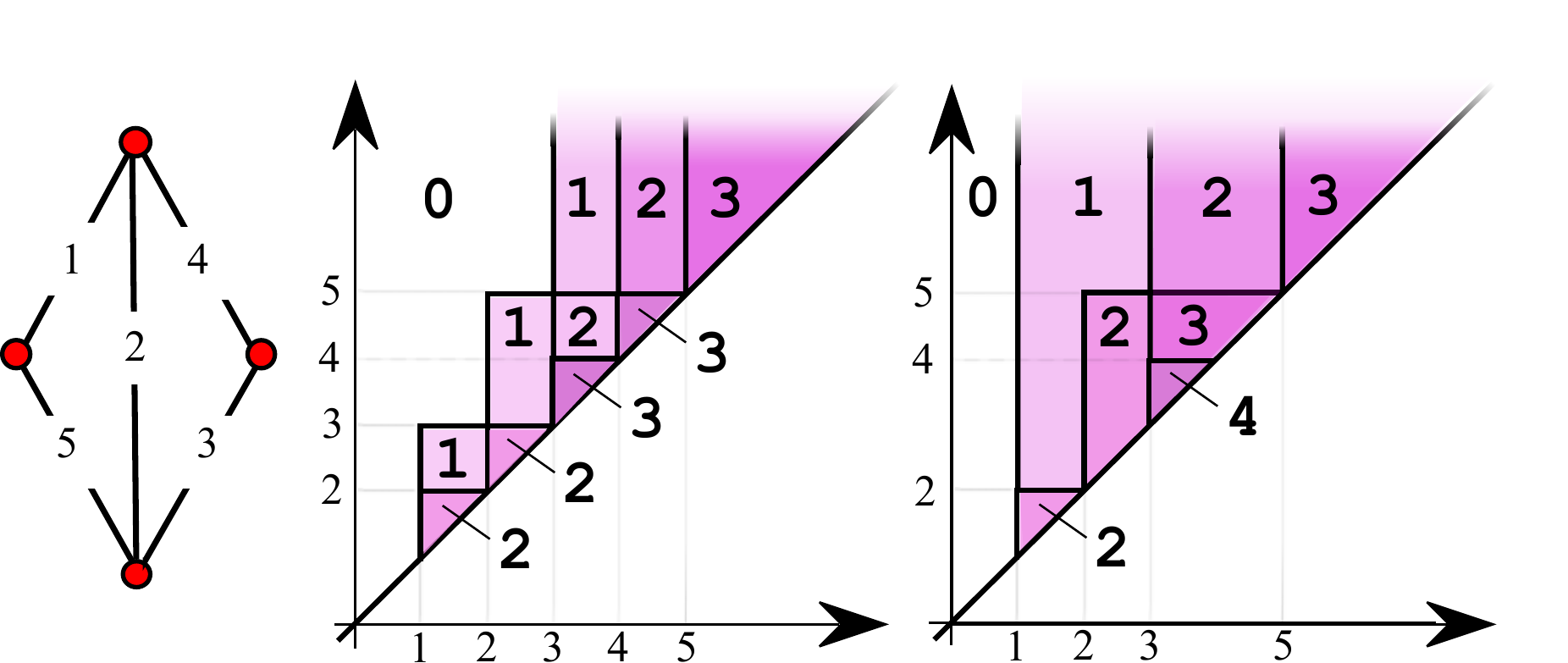}
\caption{A weighted graph $(G, f)$ (left) and its functions $\sigma^{m\mathcal{I}}_{(G, f)}$ (middle) and $\varrho^{m\mathcal{I}}_{(G, f)}$ (right).
}\label{fig:indep}
\end{figure}

\inserted{
Monotone features---although balanced---often give rise to extremely rich persistence diagrams. For this reason, it is possible to consider instead the maximal version (that could be non-balanced) of those features. In Fig.~\ref{fig:indep}, we show how maximal independent sets give rise to complex persistence diagrams, even considering as graph our running toy example (the one shown originally in Fig.~\ref{fig:theta}). For the monotone feature $\mathcal{I}$ which identifies independent sets of vertices, $m\mathcal{I}$ is not balanced (see the Appendix).}

\inserted{
Anyway, the maximal version of the feature $\mathcal{M}$, which identifies matchings, produces balanced ip-function generators (Prop.~\ref{match}). See Fig.~\ref{fig:match} for the functions $\sigma^{m\mathcal{M}}$ and $\varrho^{m\mathcal{M}}$ of the usual example of Fig.~\ref{fig:theta}.
}

\inserted{
\begin{proposition}\label{match}
The ip-function generators $\sigma^{m\mathcal{M}}$ and $\varrho^{m\mathcal{M}}$ coincide and are balanced.
\end{proposition}
\begin{proof}
If the edge set $X$ is a matching in a graph, it is a matching in all supergraphs. In a weighted graph $(G, f)$, the set of levels $w$ such that an edge set $X$ is a maximal matching in $G_w = (V_w, E_w)$ is either empty or the interval $[u_2, v_2)$ where $u_1$ is the left end-point of the $\mathcal{M}$-interval of $X$ and $v_2$ is either $+\infty$ or the left end-point of the $\mathcal{M}$-interval of a matching $Y$ containing $X$. This proves that $\sigma^{m\mathcal{M}}_{(G, f)} = \varrho^{m\mathcal{M}}_{(G, f)}$.
\newline
Let now  $g$ be another filtering function on $G$, such that $\sup_{e\in E}|f(e)-g(e)|\le h$, with $h>0$. Assume that the interval $[u_2, v_2)$ on which $X$ is a maximal matching is such that $u_2+2h < v_2 < +\infty$. Then, by Lemma~\ref{twofcts}, for the left end-point $u_3$ of the $\mathcal{M}$-interval of $X$ in $(G, g)$ and the left end-point $v_3$ of the $\mathcal{M}$-interval of $Y$ in $(G, g)$ one has $|u_2-u_3|\le h, |v_2-v_3|\le h$. So, if $X$ belongs to $S^{m\mathcal{M}}_{(G,f)}(u-h, v+h)$, it also belongs to $S^{m\mathcal{M}}_{(G, g)}(u, v)$, proving that $\sigma^{m\mathcal{M}}_{(G, f)}(u-h, v+h) \le \sigma^{m\mathcal{M}}_{(G, g)}(u, v)$.
\end{proof}
}

\begin{figure}[htb]
\centering
\includegraphics[width =0.45 \textwidth]{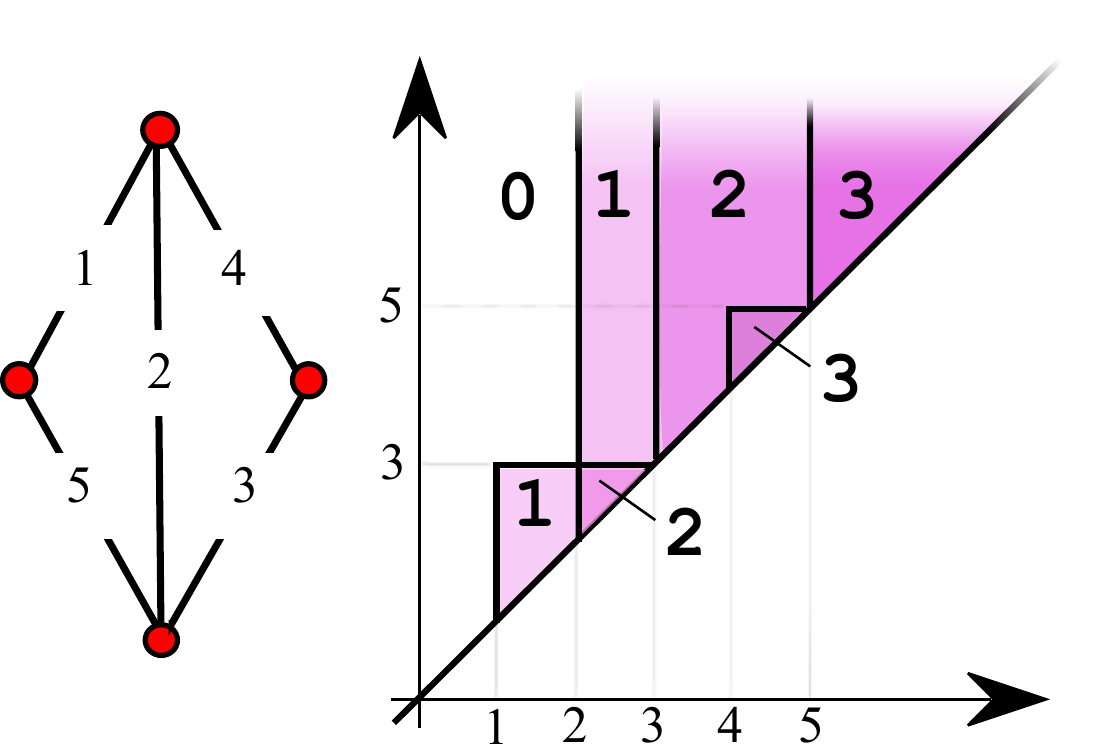}
\caption{A weighted graph $(G, f)$ (left) and its functions $\sigma^{m\mathcal{M}}_{(G, f)} = \varrho^{m\mathcal{M}}_{(G, f)}$ (right).
}\label{fig:match}
\end{figure}

\subsection{Hubs}\label{sec:hubs}

Although the informal concept of \textit{hub} is intuitively clear, it is not as easy to formalize in graph-theoretical terms. The simple idea of a vertex with (locally) maximum degree is not entirely satisfactory: in a social network it is common to find users with a lot of contacts, with whom, however, they interact poorly. Even a high sum of traffic intensities (e.g.~the number of messages exchanged between a user and their connections) is not enough to bestow a vertex the central role implied by the word \textit{hub}.

\modified{
There is an important line of research on a probabilistic concept of ``persistent hubs'' based on degree maximality \cite{dereich2009random,galashin2013existence,banerjee2021persistence} with some intersection with what we are proposing.
}

We shall use local degree prevalence as feature for building two ip-function generators: for any graph $G$ we define $\mathcal{H}:2^{V\cup E} \to \{true, false\}$ to yield $true$ only on singletons containing a vertex whose degree is greater than the ones of its neighbours. Such a vertex is called an \emph{$\mathcal{H}$-vertex} or simply a \emph{hub}.
This feature, combined with the indexing-aware persistence framework and the notion of ranging and steady feature, allows for the identification of those vertices whose role is indeed central throughout the filtration of a given weighted graph $(G, f)$.

Importantly, we preserve the flexibility granted in the realm of classical persistence: as one of the many possible variations, we could consider a vertex to be a hub if the sum of values of $f$ on the edges incident to it (instead of the degree) is greater then the sum at its neighbours.

Our proposal is to build persistence diagrams in our generalized framework, and thereafter use the selection procedure presented in~\cite{Ku16} (see \ref{sec:kurlin}) to identify relevant cornerpoints, thus identifying the ``persistent'' hubs \modified{(with a different meaning of the adjective than in  \cite{dereich2009random,galashin2013existence,banerjee2021persistence})}  of a given weighted graph.

\begin{definition}\label{sr-hub}
For any real number $w$, a vertex is a \emph{hub} (or \emph{$\mathcal{H}$-vertex}) \emph{at level} $w$ if it is an $\mathcal{H}$-vertex of the subgraph $G_w$.  It is a \emph{steady hub} (or \emph{s$\mathcal{H}$-vertex}) \emph{at} $(u, v)\in \Delta^+$ if it is an $\mathcal{H}$-vertex at all levels $w$ with $u\le w\le v$. It is a \emph{ranging hub} (or \emph{r$\mathcal{H}$-vertex}) \emph{at} $(u, v)\in \Delta^+$ if there exist levels $w \le u$ and $w'\ge v$ at which it is an $\mathcal{H}$-vertex.
\newline
$S^\mathcal{H}_{(G, f)}(u, v)$ and $R^\mathcal{H}_{(G, f)}(u, v)$ are respectively the sets of s$\mathcal{H}$-vertices and of r$\mathcal{H}$-vertices at $(u, v)$. We define $\sigma^\mathcal{H}_{(G, f)}: \Delta^+ \to \mathbb{R}$ as  the function which assigns to $(u, v) \in \Delta^+$ the number $|S^\mathcal{H}_{(G, f)}(u,v)|$ and $\varrho^\mathcal{H}_{(G, f)}: \Delta^+ \to \mathbb{R}$ as  the function which assigns to $(u, v) \in \Delta^+$ the number $|R^\mathcal{H}_{(G, f)}(u,v)|$.
\newline
We denote by  $\sigma^\mathcal{H}$ and $\varrho^\mathcal{H}$ the maps assigning $\sigma^\mathcal{H}_{(G, f)}$ and $\varrho^\mathcal{H}_{(G, f)}$ respectively to the $(\mathbb{R}, \le)$-indexed diagram corresponding to $(G, f)$.
\end{definition}

Fig.~\ref{fig:toyhub2} shows the two ip-functions $\sigma^\mathcal{H}$ and $\varrho^\mathcal{H}$ for the usual example of Fig.~\ref{fig:theta}.  Also $\sigma^\mathcal{H}$ and $\varrho^\mathcal{H}$ are not balanced (see the Appendix).

\begin{figure}[htb]
\centering
\includegraphics[width =0.7 \textwidth]{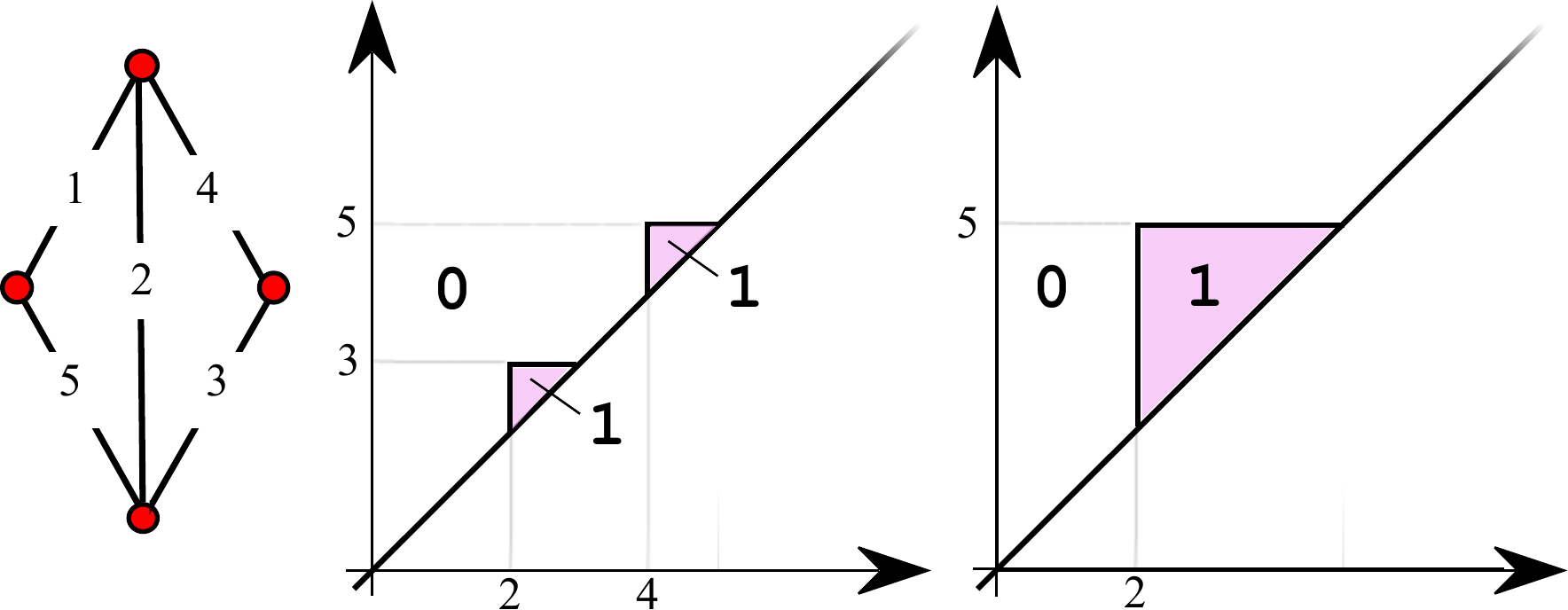}
\caption{A weighted graph $(G, f)$ (left) and its functions $\sigma^\mathcal{H}_{(G, f)}$ (middle) and $\varrho^\mathcal{H}_{(G,f)}$ (right). \modified{The topmost vertex is a hub at all levels in $[2, 3) \cup [4, 5)$.}
}\label{fig:toyhub2}
\end{figure}

\section{Persistent hubs}\label{sec:concrete}

In this Section we present a first approach to hub detection implementable on real-world graphs. We consider this work in progress a sort of exploration of the meaning of steady and ranging hubs in different contexts; however, we will not compare our results to a ground truth.

In the following examples,  instead of the functions $\sigma^\mathcal{H}_{(G, f)}$ and $\varrho^\mathcal{H}_{(G, f)}$, we will only show the corresponding persistence diagrams, to make the selection procedure clearer.

\subsection{A selection procedure}\label{sec:kurlin}

It is well-known in persistence that noise is represented by cornerpoints close to the diagonal $\Delta$. However, not all cornerpoints close to $\Delta$ necessarily represent noise, then how wide is the strip along $\Delta$ to get rid of? A smart, simple answer is offered in \cite{Ku16}, where a remarkable application to segmentation of very noisy data is given. We summarize it here for a given persistence diagram $D$.

Call {\it diagonal gap} a maximal region of the form \modified{$\{(u,v) \in \Delta^+ \, | \, a<v-a<b\}$} where no cornerpoints of $D$ lie; $b-a$ is its width. We can then form a hierarchy of diagonal gaps by decreasing width; out of it we get a hierarchy of sets of cornerpoints: We can consider the cornerpoints lying above the first, widest gap as the most relevant. Empirically, we may decide that also the cornerpoints sitting above the second, or the third widest gap are relevant, and so on. Equivalently, we consider the cornerpoints below the chosen gap to be ignored as a possible result of noise. In Fig.~\ref{fig:maxima_selection} it is possible to observe how the selection of cornerpoints above the widest diagonal gap allows to traceback those maxima (or classes of maxima depending on the multiplicity of the cornerpoints), that are more relevant with respect to the trend of the time series.

\begin{figure}[tbh]
\centering
\includegraphics[width=1\textwidth]{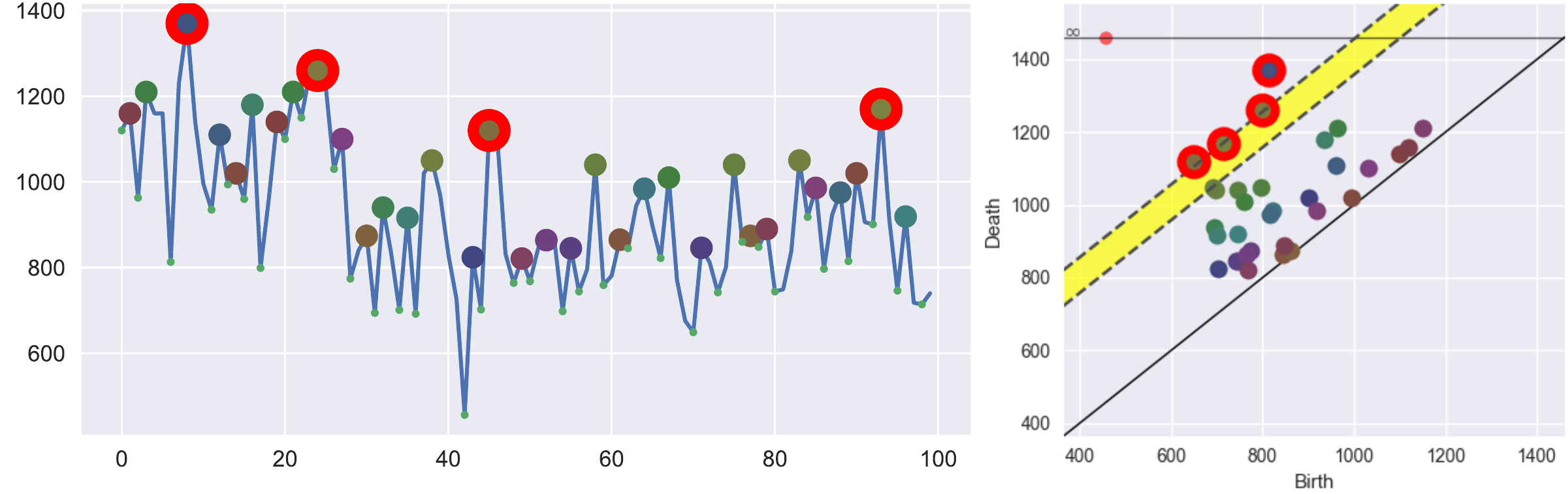}
\caption{Selecting maxima in a time series. \textbf{Left.} Flow of the Nile from $1871$ to $1970$. Data freely available at \href{https://vincentarelbundock.github.io/Rdatasets/datasets.html}{vincentarelbundock.github.io}. \textbf{Right.} Cornerpoints selected by considering the widest diagonal gap (in yellow).}\label{fig:maxima_selection}
\end{figure}

In the next Sections we apply this selection criterion to the persistence diagrams corresponding to the functions $\sigma^\mathcal{H}_{(G, f)}$ and $\varrho^\mathcal{H}_{(G, f)}$, computed for some networks and some filtering functions. The vertices identified by the so selected cornerpoints will be called {\it persistent hubs}, in particular {\it persistent steady hubs} or {\it persistent ranging hubs}.

\subsection{Airports}

A first attempt of the search for relevant hubs has been realized on a set of 44 major North-American cities (41 in the US, three in Canada; the ones in capital letters in the Amtrak railway map; see Table~\ref{tab:airports}). The edges connect cities between which there have been flights in a randomly chosen but fixed week (June 11 to 17, 2018). Flight data have been obtained from Google Flights by selecting direct flights with Business Class; distances have been found at \href{https://www.prokerala.com/}{Prokerala.com}. A single vertex has been considered for each city with more than one airport.

\begin{table}[tbh]
\begin{center}
\begin{tabular}{llllll}
\toprule
\multicolumn{4}{c}{\textbf{Vertices (degree)}}\\
\midrule
Albuquerque (13) & Atlanta (42) & Baltimore (16) & Boston (30) \\ Buffalo (8) &
Cheyenne (0) & Chicago (40)& Cincinnati (19) \\ Cleveland (13) & Dallas (41) &
Denver (39) & Detroit (35) \\ El Paso (7) & Houston (40) & Indianapolis (17) &
Jacksonville (12) \\ Kansas City (19) & Las Vegas (23) & Los Angeles (37) & Memphis (11) \\
Miami (30) & Milwaukee (14) & Mobile (3) & Montreal (16) \\ New Orleans (16)&
New York (35)& Oakland/Emeryville (7) & Philadelphia (34) \\ Phoenix (35) & Pittsburgh (14) & Portland (25) & Sacramento (16)\\ Salt Lake City (33) & San Antonio (17)& San Diego (26) & San Francisco (35)\\ Seattle (34) & St. Louis (17) & St. Paul-Minneapolis (38) & Tampa (19) \\ Toronto (26) & Tucson (10)& Vancouver (18) & Washington (32)\\
\bottomrule
\end{tabular}
\caption{The towns considered as vertices and the respective degrees in the graph.}\label{tab:airports}
\end{center}
\end{table}

As filtering functions we used:

\begin{itemize}

\item distance
\item number of flights in the fixed week
\item their product

\end{itemize}

\noindent and their opposites (+their maximum). For each such choice we looked for steady and ranging hubs, for a total of twelve different persistence diagrams. Note that the same vertex can contribute to several cornerpoints of the persistence diagram of $\sigma^\mathcal{H}_{(G, f)}$, whereas this cannot happen for $\varrho^\mathcal{H}_{(G, f)}$.

Next, we report results in which where the interest resides in the identification of hubs which do not rank very high by their degree. In particular, we do not find of particular interest that Atlanta, Dallas, Chicago and Houston turn out to be often persistent ranging or steady hubs, since they have the highest degrees in the graph (42, 41, 40 and 40 respectively).

\begin{figure}[tbh]
\centering
\includegraphics[width=65mm]{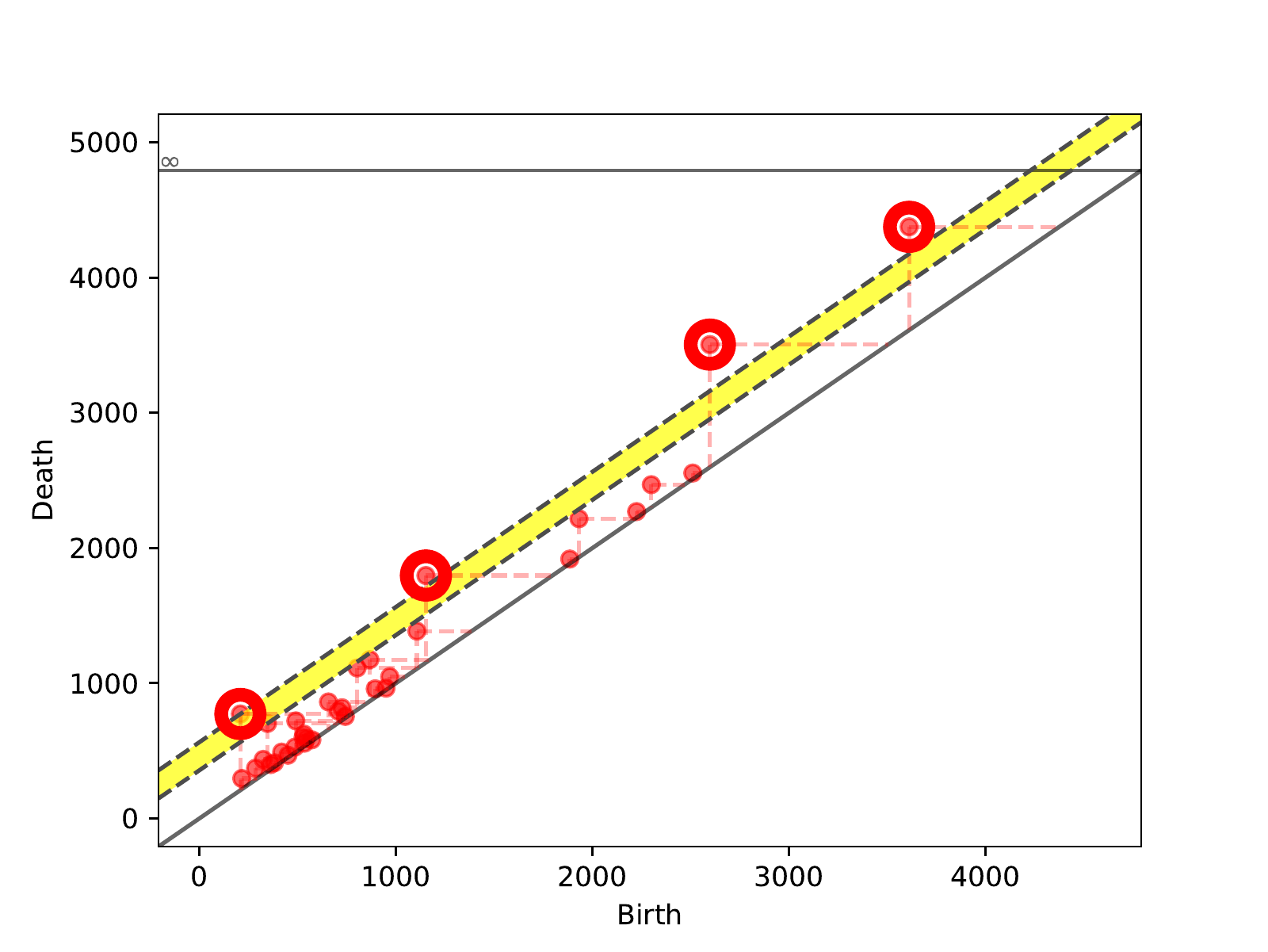}
\caption{Filtering function: distance; steady hubs. Persistent steady hubs above the widest diagonal gap: two cornerpoints represent Atlanta, one Dallas and one Seattle.}\label{fig:steadydist}
\end{figure}

The first occurrence of a persistent hub which is rather far from having highest degrees is with the filtering function distance: Seattle is just twelfth in the degree rank, but appears above the widest diagonal gap as a steady hub (Figure~\ref{fig:steadydist}). Persistent steady hubs are: Atlanta (with two cornerpoints), Dallas, Seattle.

Surprisingly, if we use the opposite of distance (summed to the maximum distance, for ease of representation), the cornerpoints corresponding to vertices with highest degrees are located under the widest diagonal gap  (Figure~\ref{fig:steadyminusdist}).
Persistent steady hubs are: Los Angeles, San Francisco, Seattle.

\begin{figure}[tbh]
\centering
\includegraphics[width=65mm]{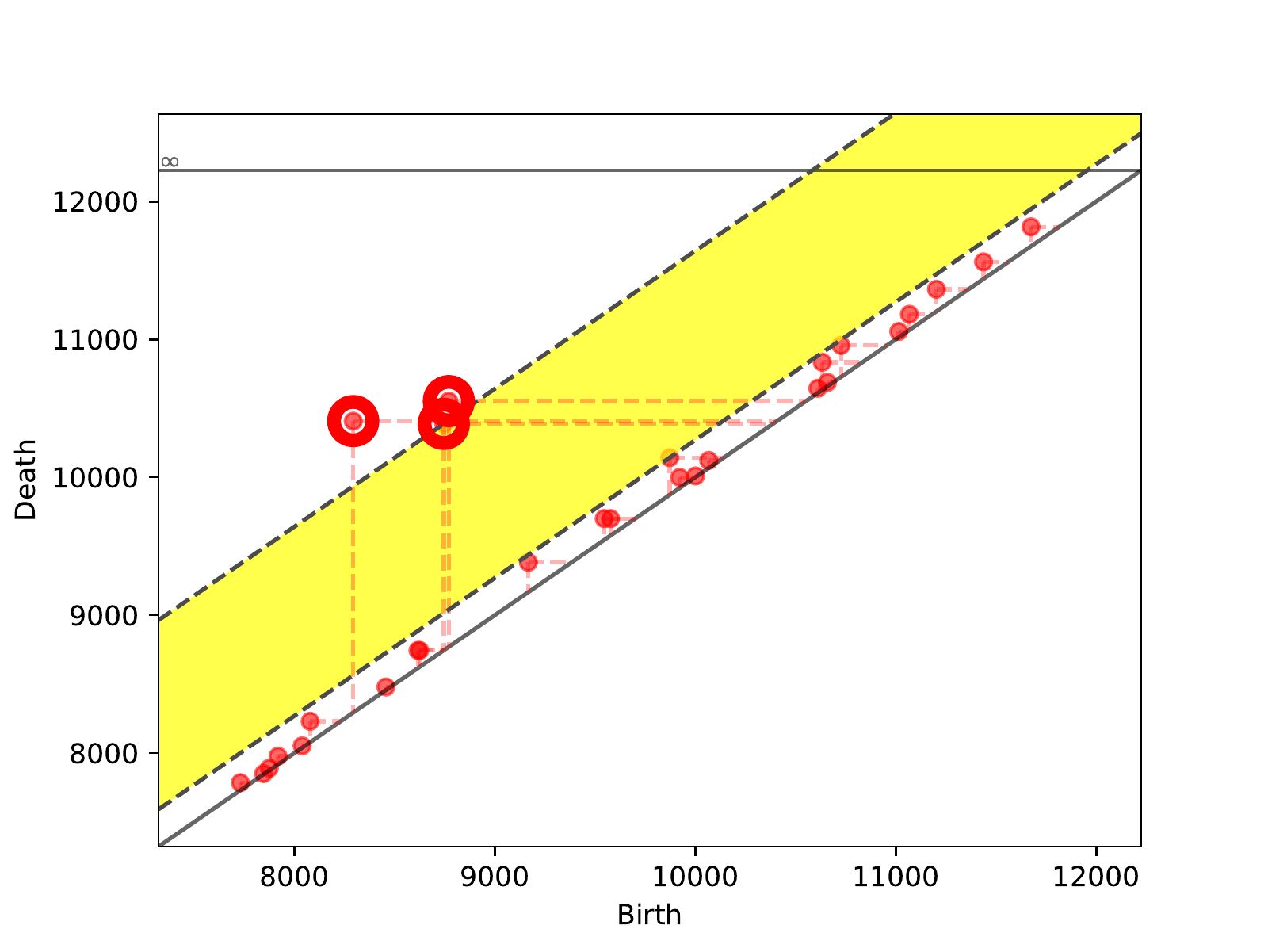}
\caption{Filtering function: max distance minus distance; steady hubs. Persistent steady hubs above the widest diagonal gap: Los Angeles, San Francisco, Seattle.}\label{fig:steadyminusdist}
\end{figure}

New York City has the eighth highest degree (35, together with Detroit, Phoenix and San Francisco). Still, we would expect it to appear as a hub, in the common sense of the term. In fact, it occurs as one of the few ranging hubs when the filtering functions (max minus number of flights) and distance$\cdot$(max minus number of flights) are used.
\newline
Ranging hubs for (max minus number of flights): Atlanta, Chicago, Dallas, New York.
\newline
Ranging hubs for the product filtering function are Atlanta, Chicago, Dallas, New York, Vancouver.

\subsection{Characters co-occurrence in a novel}

A classical benchmark for the analysis of hubs in co-occurrence graphs is given by {\it Les Mis\'erables}. The network representing the co-occurrence of its characters is freely available at \href{https://github.com/graphistry/pygraphistry/blob/master/demos/data/lesmiserables.csv}{Graphistry}. The graph has 77 major characters as vertices; each of the 254 edges joins two characters which appear together in at least one scene; the weight on an edge is the number of common occurrences. We used the inverse of the weight as a filtering function. We compare our results with the ones of \cite{rieck2018clique}, where the notion of {\it clique-community centrality} was used to spot particularly important characters: Table~\ref{tab:ccc}.

\begin{table}[tb]
\begin{center}
\begin{tabular}{l l l}

\multicolumn{3}{c}{\textbf{Steady hubs}}\\


Cosette & Courfeyrac & Enjolras \\
Marius & Myriel & Valjean  \\
\\
\multicolumn{3}{c}{\textbf{Ranging hubs}}\\


Cosette & Courfeyrac & Enjolras \\
Marius & Myriel & Valjean \\
\\
\multicolumn{3}{c}{\textbf{Clique-community centrality}}\\

Enjolras & Fantine & Gavroche \\
Marius & Valjean & \\

\end{tabular}
\caption{Hubs in Les Mis\'erables characters co-occurrence. Comparing results obtained via the steady and ranging persistence construction and clique-community centrality.} \label{tab:ccc}
\end{center}
\end{table}

Our method spots Cosette as a hub, whereas clique-community centrality does not. On the contrary, our technique misses Gavroche and Fantine. Both methods miss Javert. We are particularly puzzled by the result of Kurlin's selection method: above the second widest diagonal gap (the first obviously isolates Jean Valjean) we find only Enjolras.

\inserted{
\subsection{Time-varying hubs}
}
\inserted{Weighted graphs can represent discrete dynamics in time-varying process. It is possible to keep track of persistence hubs obtaining a concise representation of the relative importance of each hub in time. We considered the characters co-occurrence in five subsequent books of the Game of Thrones saga, and applied the algorithm mentioned above for the analysis of character co-occurrence in Les Mis\'erables. In this case, however, characters evolve throughout the books. A global analysis, i.e., computing hubs on the graphs obtained considering summary statistics on the five book hardly carries dynamical information. On the contrary, persistence hubs yield an easily visualizable summary of the characters' roles in time. See~Fig.~\ref{fig:got}}

\begin{figure}[tbh]
    \centering
    \includegraphics[width=1\textwidth]{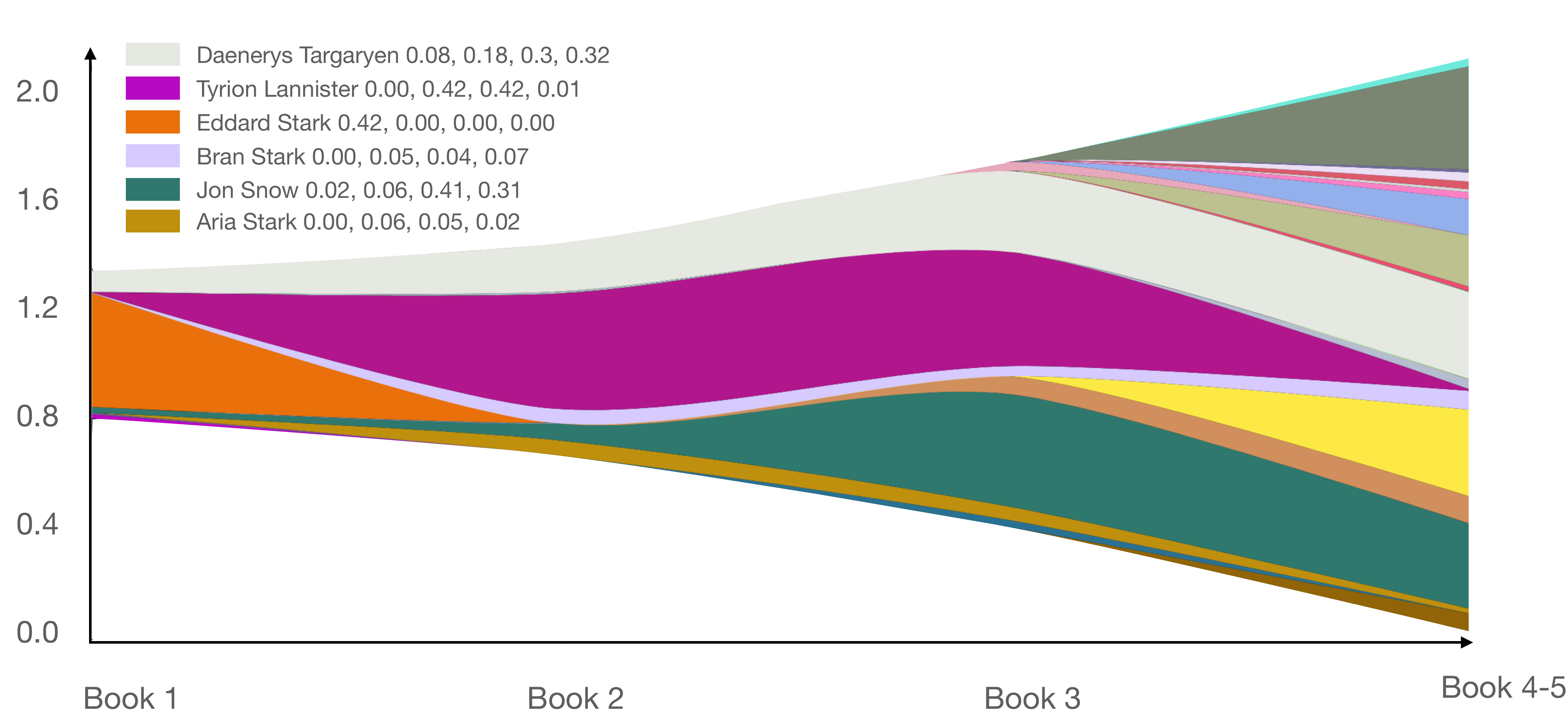}
    \caption{Evolution of Game of Thrones hub characters throughout five books. The legend reports the first six hubs and their persistence values per book.}\label{fig:got}
\end{figure}

\subsection{Languages}

The website \href{http://test.terraling.com/groups/7}{TerraLing.com} contains much information, consisting of 165 properties, about several languages. It was used in an interesting research \cite{port2018persistent} on persistent cycles in language families. Unfortunately the amount of information varies quite a lot from language to language. We analysed the mutual relations of 19 languages (18 of the European Union plus Turkish: Table~\ref{tab:ling}) for which at least 50\% of the 165 properties are checked. The graph is the complete one with 19 vertices. The filtering function defined on each edge is the opposite of the normalised quantity of common properties of the two languages that it connects. Ranging and steady hubs coincide and are: Castilian, Catalan, Dutch, English, Portuguese, Swedish.

\begin{table}[tbh]
\begin{center}
\begin{tabular}{lllll}
\hline
\multicolumn{5}{c}{\textbf{Languages}}\\
\hline
Castilian & Catalan & Czech & Croatian & Danish  \\
Dutch& English &Finnish & French & Galician  \\
German & Greek & Hungarian & Italian& Polish \\ Portuguese & Romanian &
Swedish  & Turkish  &  \\
\hline
\end{tabular}
\caption{The 19 considered languages.}\label{tab:ling}
\end{center}
\end{table}

Apart from the presence of English, which might also be biased by the great quantity of information available, we have no key for interpreting these results. For this and for the previous applications, we would very much like to set up a research with specific experts.

\inserted{
\section{Digraph persistence}\label{sec:digraph}
}
\inserted{
In this section, let $(G, f)$, with $G=(V, A)$, be any weighted digraph. Given a {\em feature} $\mathcal{F}: 2^{V\cup A} \to \{true, false\}$, it is straightforward to extend the definitions of {\em balanced} ip-function (Def.~\ref{def:balfct}), of {\em natural pseudodistance} (Def.~\ref{def:natural}), the stability theorem (Thm.~\ref{thm:stab}) and the definitions of {\em steady} and {\em ranging} sets (Def.~\ref{def:sr}) and of the ip-function generators $\sigma^\mathcal{F}$ and $\varrho^\mathcal{F}$ (Def.~\ref{def:rhosigma}, Prop.~\ref{supergp})  to this setting.
}

\begin{figure}
    \centering
    \begin{subfigure}[t]{\linewidth}
        \centering
        \includegraphics[width =0.75 \textwidth]{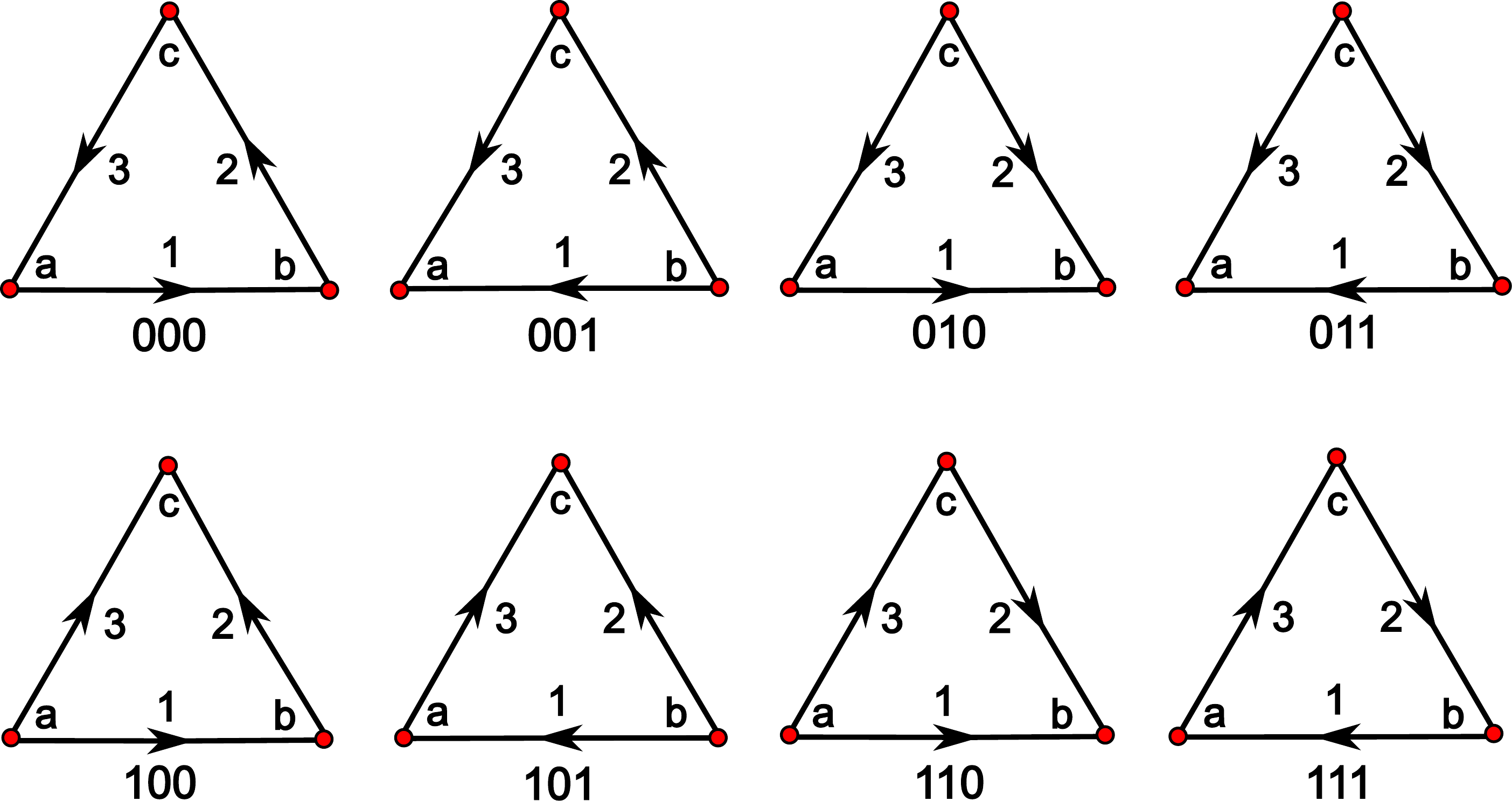}
        \caption{The eight tournaments on three vertices, with $\{1, 2, 3\}$-valued filtering functions.}\label{fig:triangles}
    \end{subfigure}

    \begin{subfigure}[t]{\linewidth}
        \centering
        \includegraphics[width =.75 \textwidth]{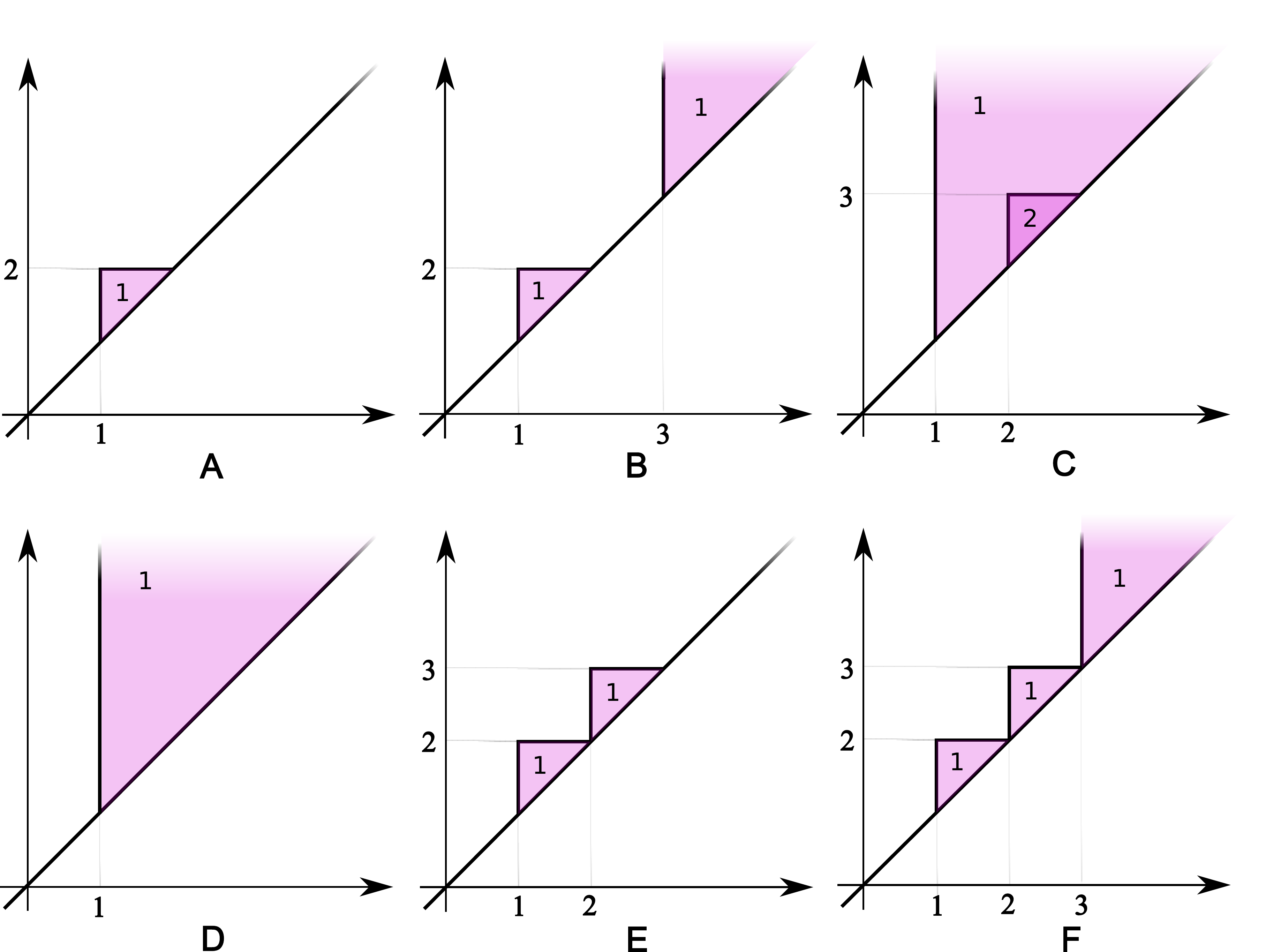}
        \caption{The ip-functions corresponding to the digraphs of Fig.~\ref{fig:triangles} with respect to features $\mathcal{DH}$ and $\mathcal{K}$.}\label{fig:diagrams}
    \end{subfigure}
    \caption{Examples of digraphs and ip-functions; for the correspondence see Table~\ref{tab:h} and Table~\ref{tab:k}.}
\end{figure}

\inserted{
We define $\mathcal{DH}: 2^{V\cup A} \to \{true, false\}$ to yield $true$ only on singletons containing a vertex whose outdegree is greater than the ones of its neighbours. Also in this case, there are many possible variations of this feature: we recover the notions of {\em hub}, {\em steady hub} and {\em ranging hub} and ip-function generators $\sigma^{\mathcal{DH}}$ and $\varrho^{\mathcal{DH}}$ as in Section~\ref{sec:hubs}.
}

\inserted{
Fig.~\ref{fig:triangles} presents all tournaments on three vertices, with injective functions with values in the set $\{1, 2, 3\}$. Fig.~\ref{fig:diagrams} shows the values of some ip-functions. The correspondence between weighted tournaments and functions is given in Table~\ref{tab:h}. On these digraphs, $\sigma^{\mathcal{DH}}$ and $\varrho^{\mathcal{DH}}$ yield coinciding functions. \inserted{However, this is not always the case, as shown in Fig.~\ref{fig:toydir}.}
\begin{table}[tbh]
\begin{center}
\begin{tabular}{|c||c|c|c|c|c|c|c|c|}
\hline
\multicolumn{9}{|c|}{\textbf{Hubs}}\\
\hline
&000 & 001 & 010 & 011 & 100 & 101 & 110 & 111 \\
\hline
$\sigma^{\mathcal{DH}}_{(G, f)} = \varrho^{\mathcal{DH}}_{(G, f)}$  & A & B & C & C & D & D & A & B \\
\hline
\end{tabular}
\caption{The correspondence between the weighted digraphs of Fig.~\ref{fig:triangles} and the diagrams of Fig.~\ref{fig:diagrams} for feature ${\mathcal{DH}}$.}\label{tab:h}
\end{center}
\end{table}\label{tab:h}
}

\begin{figure}[htb]
\centering
\includegraphics[width =0.6 \textwidth]{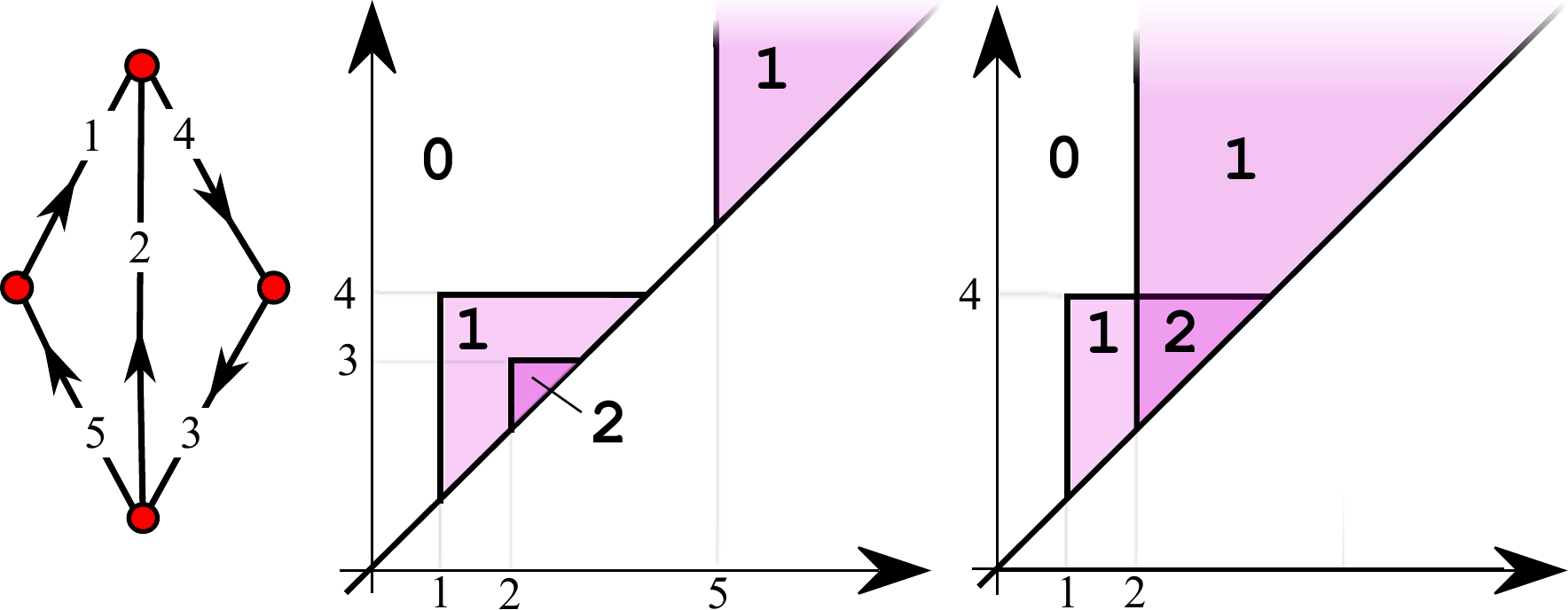}
\caption{A weighted digraph $(G, f)$ (left) and its functions $\sigma^{\mathcal{DH}}_{(G, f)}$ (middle) and $\varrho^{\mathcal{DH}}_{(G, f)}$ (right).
}\label{fig:toydir}
\end{figure}

\inserted{
There are two opposite definitions of a {\em kernel} of a digraph; we shall consider the one given in~\cite{morgenstern1953theory}. However, alternative definitions (see, e.g.,~\cite{galeana2014existence}) give also rise to admissible features in our framework. We define the feature $\mathcal{K}: 2^{V \cup A} \to \{true, false\}$ to yield $true$ only on kernels, i.e. independent sets $X$ of vertices such for every vertex $w \in V-X$, there exists at least one arc $a \in A$ with $w$ as tail and head in $X$, where independence is defined with respect to the underlying undirected graph. Then $\sigma^\mathcal{K}$ and $\varrho^\mathcal{K}$ are ip-function generators. The correspondence between weighted tournaments and functions is given in Table~\ref{tab:k}
\begin{table}[tbh]
\begin{center}
\begin{tabular}{|c||c|c|c|c|c|c|c|c|}
\hline
\multicolumn{9}{|c|}{\textbf{Kernels}}\\
\hline
&000 & 001 & 010 & 011 & 100 & 101 & 110 & 111 \\
\hline
$\sigma^\mathcal{K}_{(G, f)}$   & E & F & D & F & F & F & D & E \\
\hline
$\varrho^\mathcal{K}_{(G, f)}$   & E & C & D & C & F & F & D & E \\
\hline
\end{tabular}
\caption{The correspondence between the weighted digraphs of Fig.~\ref{fig:triangles} and the diagrams of Fig.~\ref{fig:diagrams} for feature $\mathcal{K}$.}\label{tab:k}
\end{center}
\end{table}\label{tab:k}
}

\inserted{None of the ip-function generators $\sigma^{\mathcal{DH}}$, $\varrho^{\mathcal{DH}}$, $\sigma^\mathcal{K}$, $\varrho^\mathcal{K}$ is balanced (see the Appendix).}

\section{Conclusions}

We  introduced ip-functions in a fairly general setting and studied their stability. We have then restricted our scope to the categories of graphs and digraphs, where we have defined steady and ranging sets according to features relative to the given (di)graphs. 

\modified{
    We showed how graph-theoretical features can be used directly to obtain a concise representation of weighted undirected and directed graphs as persistence diagrams. In particular, we believe that the steady and ranging ip-function generators allow for a more streamlined analysis of graphs and networks bypassing the construction of auxiliary simplicial complexes. Although the steady and ranging sets yield equivalent results in some cases, persistence diagrams associated with ranging sets are generally simpler than the ones derived from steady sets, so the information is represented in a more condensed way. This is not the only reason for considering both representations.}
\modified{
    In our applications, we focused on the notion of hub. There, we showcased how the ranging representation of hubs is relevant for hub detection: a vertex might be relevant for the global dynamics of a network if it has local degree prevalence at far enough levels. For example, in a graph whose vertices represent users of a social network, edges represent ``friendship'', and weights represent geographical distance, we conjecture that high-persistence ranging hubs might be crucial for the diffusion of ``viral'' documents. Analogously, we thought that an airport might have a key role if it has a sort of centrality both at a regional and international level, but not necessarily at all intermediate ones. 
}

\section*{Acknowledgments}
We are indebted to Diego Alberici, Emanuele Mingione, Pierluigi Contucci, Patrizio Frosini, Lorenzo Zuffi and above all Pietro Vertechi for many fruitful discussions. Article written within the activity of INdAM-GNSAGA. We thank the reviewers for the very helpful comments and suggestions. On behalf of all authors, the corresponding author states that there is no conflict of interest.

\begin{figure}[htb]
\centering
\includegraphics[width = 0.65\textwidth]{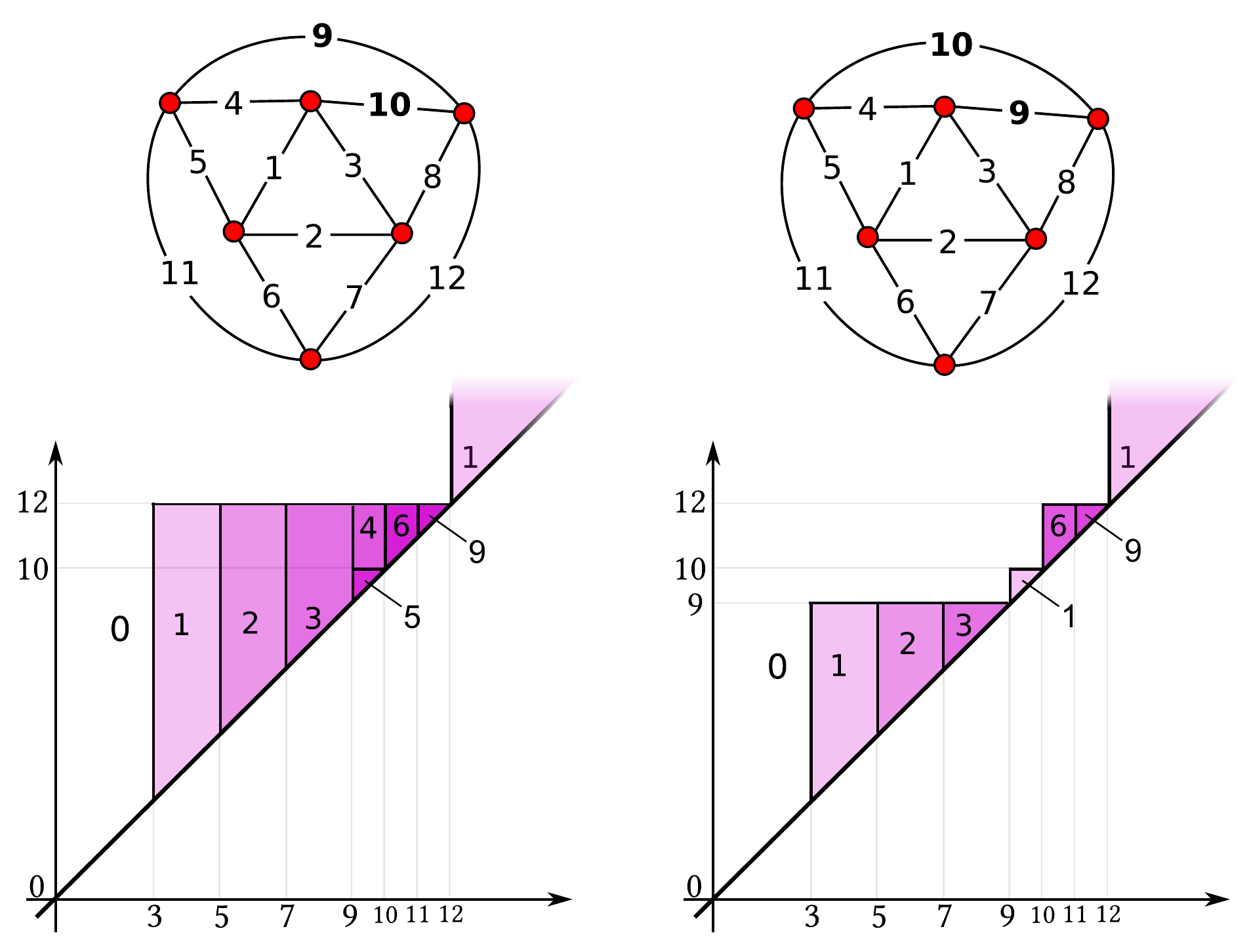}
\caption{$\sigma^\mathcal{EU}$ is not balanced: filtering function $f$ left, $g$ right.
}\label{fig:sunstable}
\end{figure}

\section*{Appendix: Unbalanced}

In order to show that some of the proposed ip-functions are not balanced---so their persistence diagrams do not generally enjoy stability---we give examples which do not respect Def.~\ref{def:balfct}.

The ip-function generator $\sigma^\mathcal{EU}$ is not balanced, as the example of Fig.\ref{fig:sunstable} shows: in fact, the maximum absolute value of the weight difference on the same edges is 1, and $\sigma^\mathcal{EU}_{(G, f)}(4.5-1, 10+1) = 1 > 0 = \sigma^\mathcal{EU}_{(G, g)}(4.5, 10)$.

\begin{figure}[htb]
\centering
\includegraphics[width = 0.65\textwidth]{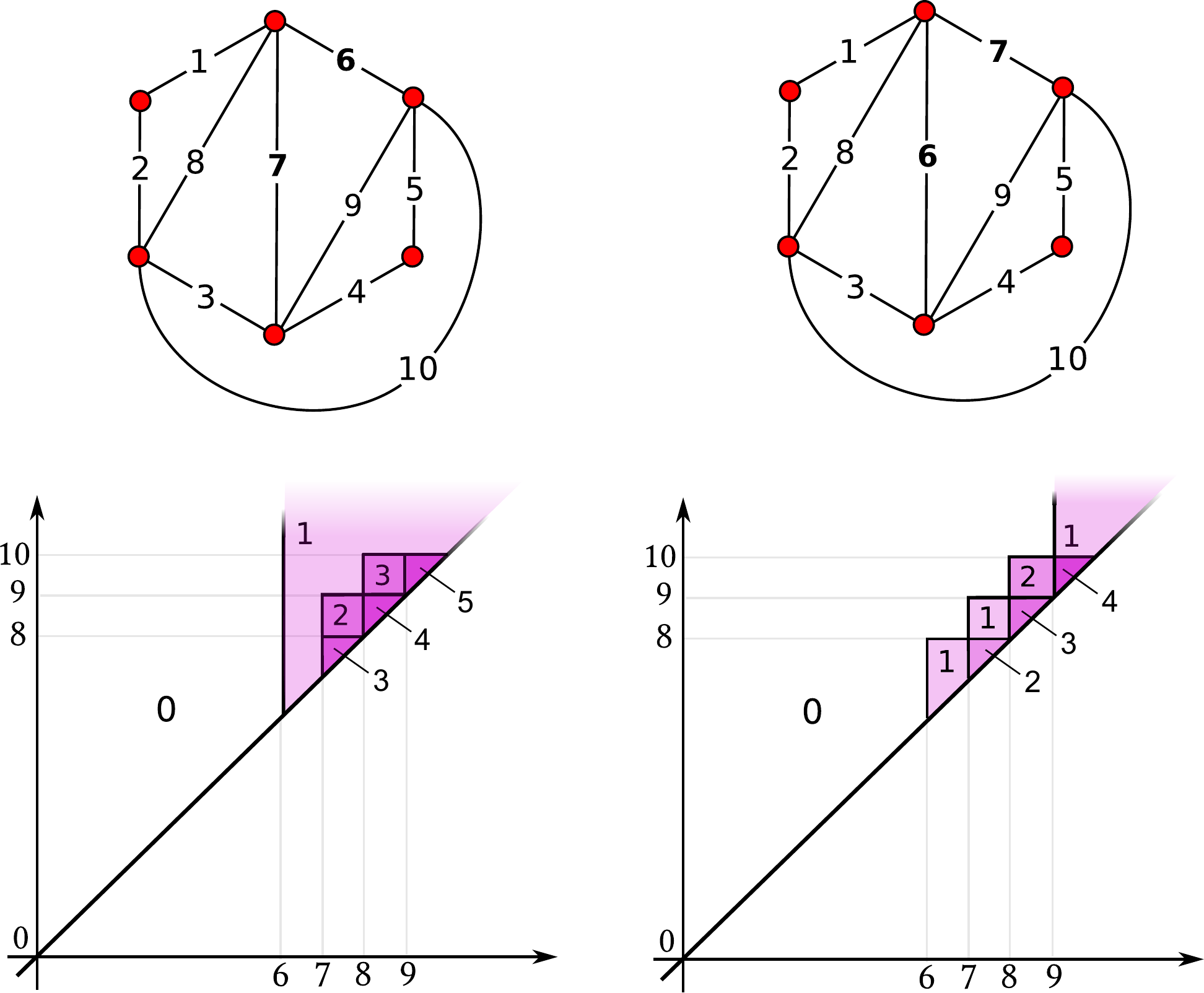}
\caption{$\varrho^\mathcal{EU}$ is not balanced: filtering function $f$ left, $g$ right.
}\label{fig:runstable}
\end{figure}

Also the ip-function generator $\varrho^\mathcal{EU}$ is not balanced, as the example of Fig.\ref{fig:runstable} shows: in fact, the maximum absolute value of the weight difference on the same edges is 1, and $\varrho^\mathcal{EU}_{(G, f)}(7.5-1, 10+1) = 1 > 0 = \varrho^\mathcal{EU}_{(G,g)}(7.5, 10)$.

\begin{figure}[htb]
\centering
\includegraphics[width = 0.35\textwidth]{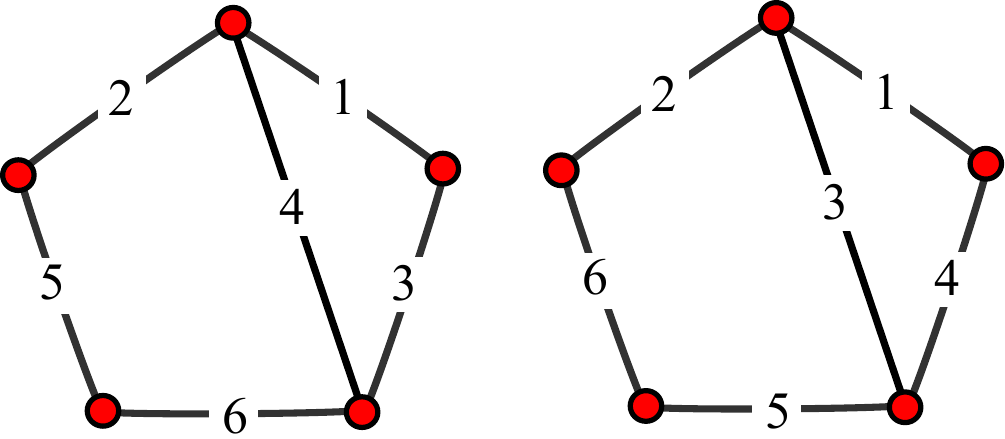}
\caption{$\sigma^{m\mathcal{I}}$ is not balanced: filtering function $f$ left, $g$ right.
}\label{fig:unst-ind-s}
\end{figure}

The ip-function generator $\sigma^{m\mathcal{I}}$ is not balanced: for the two filtering functions on the graph of Fig.~\ref{fig:unst-ind-s} the maximum difference in absolute value on the same edges is 1, but  $\varrho^{m\mathcal{I}}_{(G, f)}(3.5-1, 6+1) = 1 > 0 = \varrho^{m\mathcal{I}}_{(G,g)}(3.5, 6)$.

\begin{figure}[htb]
\centering
\includegraphics[width = 0.3\textwidth]{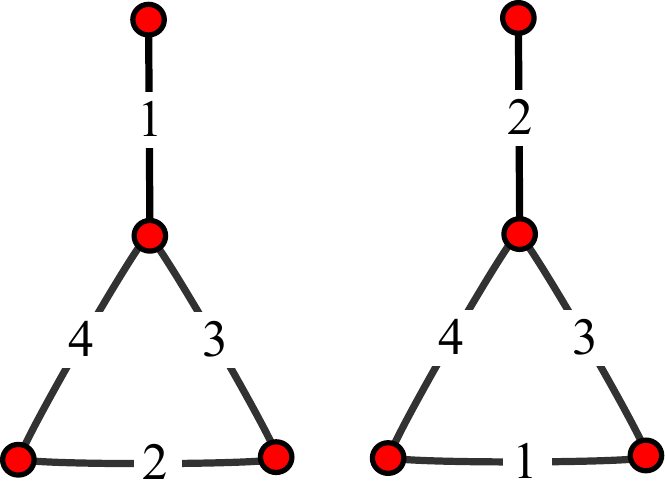}
\caption{$\varrho^{m\mathcal{I}}$ is not balanced: filtering function $f$ left, $g$ right.
}\label{fig:unst-ind-r}
\end{figure}

The ip-function generator $\varrho^{m\mathcal{I}}$ is not balanced: for the two filtering functions on the graph of Fig.~\ref{fig:unst-ind-r} the maximum difference in absolute value on the same edges is 1, but  $\varrho^{m\mathcal{I}}_{(G, f)}(3.5-1, 5+1) = 3 > 2 = \varrho^{m\mathcal{I}}_{(G,g)}(3.5, 5)$.

$\sigma^\mathcal{H}$ is not a balanced ip-function generator, as the example of Fig.~\ref{s-unstable} shows: the maximum absolute value of the weight difference on the same edges is 2, but $ \sigma^\mathcal{H}_{(G, f)}(4-2, 9+2) = 1 > 0 = \sigma^\mathcal{H}_{(G, g)}(4,9)$.

\begin{figure}[htb]
    \centering
    \begin{subfigure}[b]{.3\textwidth}
           \centering
           \includegraphics[width=\textwidth]{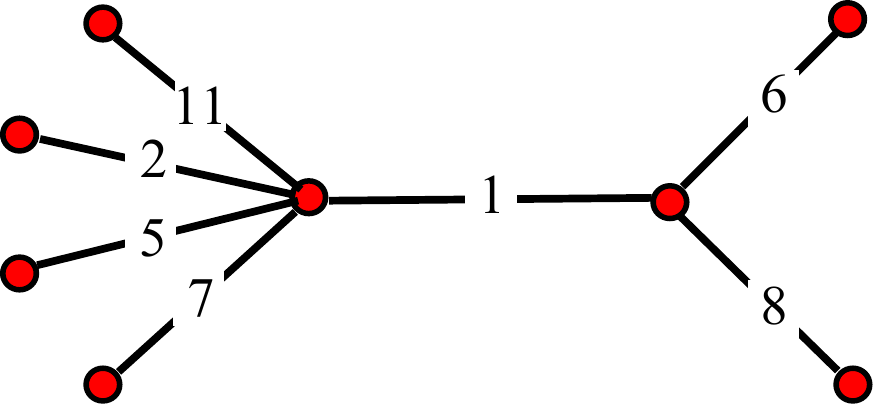}
    \end{subfigure}
    \begin{subfigure}[b]{.3\textwidth}
           \centering
           \includegraphics[width=\textwidth]{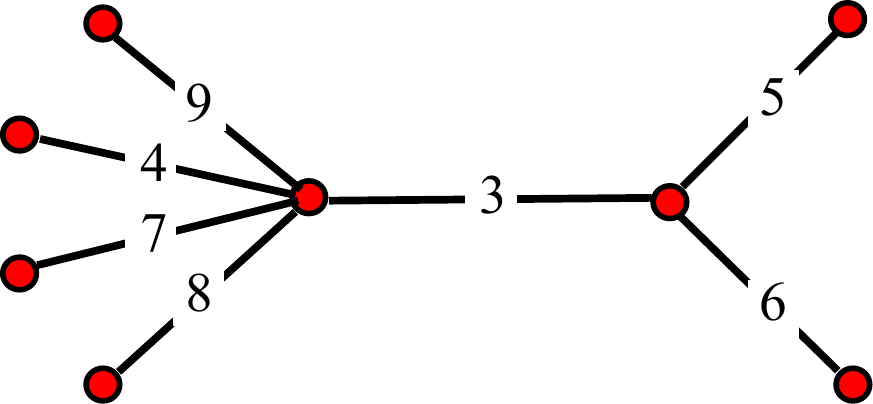}
   \end{subfigure}
    \caption{$\sigma^\mathcal{H}$ is not balanced: filtering function $f$ left, $g$ right.}\label{s-unstable}
    \end{figure}

There are counterexamples which are even simpler than this and the one of Fig.~\ref{r-unstable}. These have the advantage to hold also if ``$>$'' is substituted by``$\ge$'' in the definition of hub (what we don't think to be a good idea).

\begin{figure}[htb]
    \centering
    \begin{subfigure}[b]{.3\textwidth}
           \centering
           \includegraphics[width=\textwidth]{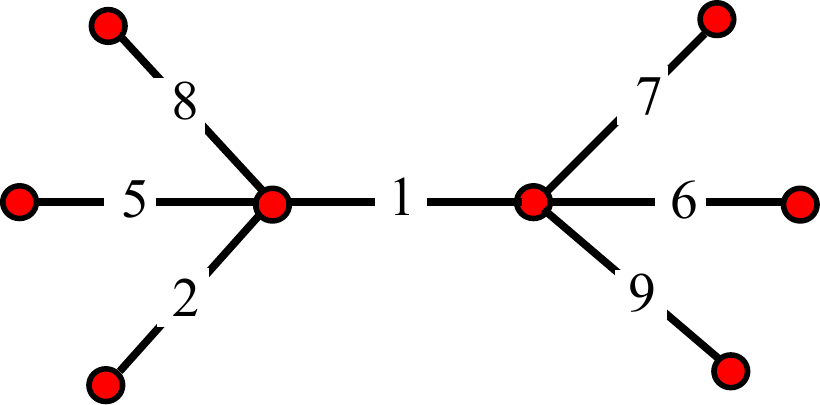}
              \end{subfigure}
    \begin{subfigure}[b]{.3\textwidth}
           \centering
           \includegraphics[width=\textwidth]{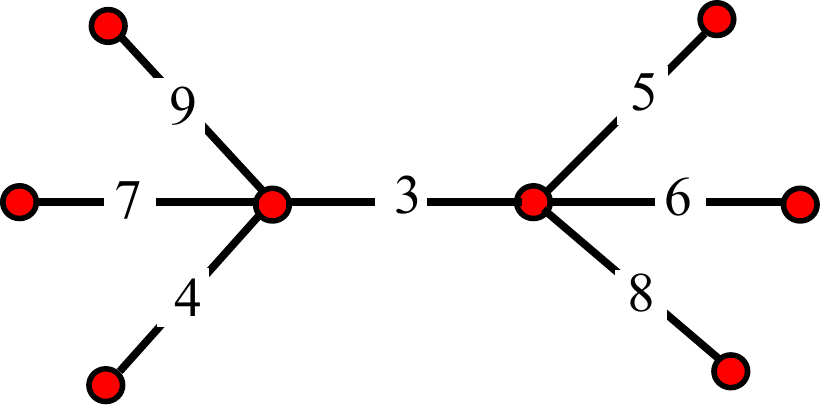}
    \end{subfigure}
    \caption{$\varrho^\mathcal{H}$ is not balanced: filtering function $f$ left, $g$ right.}\label{r-unstable}
    \end{figure}

Also $\varrho^\mathcal{H}$ is not a balanced ip-function generator, as the example of Fig.~\ref{r-unstable} shows:  the maximum absolute value of the weight difference on the same edges is 2, but  $ \varrho^\mathcal{H}_{(G, f)}(5-2, 6+2) = 1 > 0 = \varrho^\mathcal{H}_{(G, g)}(5,6)$.

\inserted{
In order to show that $\sigma^{\mathcal{DH}}$ and $\varrho^{\mathcal{DH}}$ are not balanced, consider the weighted tournaments 010 as $(G, f)$ and 011 as $(G', f')$. For the isomorphism $\psi$ which swaps vertices $a$ and $b$, one has $|f(e)-f'\big(\psi(e)\big)| \le 1$ for all $e \in A$, but $\sigma^{\mathcal{DH}}_{(G, f)}(2.5-1, 3+1) = \varrho^{\mathcal{DH}}_{(G, f)}(2.5-1, 3+1) = 1 > 0 = \varrho^{\mathcal{DH}}_{(G', f')}(2.5, 3) = \sigma^{\mathcal{DH}}_{(G', f')}(2.5, 3)$.
}

\inserted{
The ip-function generator $\sigma^\mathcal{K}$ is not balanced: consider the weighted tournaments 010 as $(G, f)$ and 011 as $(G', f')$. For the isomorphism $\psi$ which swaps vertices $a$ and $b$, one has $|f(e)-f'\big(\psi(e)\big)| \le 1$ for all $e \in A$, but $\sigma^{\mathcal{K}}_{(G, f)}(2.5-1, 4+1) = 1 > 0 = \sigma^{\mathcal{K}}_{(G', f')}(2.5, 4)$.
}

\inserted{
Finally, also $\varrho^\mathcal{K}$ is not a balanced ip-function generator: consider the weighted tournaments 001 as $(G, f)$ and 101 as $(G', f')$. For the isomorphism $\psi$ which swaps vertices $a$ and $c$, one has $|f(e)-f'\big(\psi(e)\big)| \le 1$ for all $e \in A$, but $\varrho^{\mathcal{K}}_{(G, f)}(2.5-1, 4+1) = 1 > 0 = \varrho^{\mathcal{K}}_{(G', f')}(2.5, 4)$.
}

\bibliographystyle{spmpsci}
\bibliography{FerriGraPer}

\begin{thebibliography}{10}
\providecommand{\url}[1]{{#1}}
\providecommand{\urlprefix}{URL }
\expandafter\ifx\csname urlstyle\endcsname\relax
  \providecommand{\doi}[1]{DOI~\discretionary{}{}{}#1}\else
  \providecommand{\doi}{DOI~\discretionary{}{}{}\begingroup
  \urlstyle{rm}\Url}\fi

\bibitem{alon2008every}
Alon, N., Shapira, A.: Every monotone graph property is testable.
\newblock SIAM Journal on Computing \textbf{38}(2), 505--522 (2008)

\bibitem{vijay2020weighted}
Anand, D.V., Meng, Z., Xia, K., Mu, Y.: Weighted persistent homology for
  osmolyte molecular aggregation and hydrogen-bonding network analysis.
\newblock Scientific Reports (Nature Publisher Group) \textbf{10}(1) (2020)

\bibitem{banerjee2021persistence}
Banerjee, S., Bhamidi, S.: Persistence of hubs in growing random networks.
\newblock Probability Theory and Related Fields pp. 1--63 (2021)

\bibitem{bergomi2021beyond}
Bergomi, M.G., Ferri, M., Vertechi, P., Zuffi, L.: Beyond topological
  persistence: Starting from networks.
\newblock Mathematics \textbf{9}(23) (2021).
\newblock \doi{10.3390/math9233079}.
\newblock \urlprefix\url{https://www.mdpi.com/2227-7390/9/23/3079}

\bibitem{bergomi2020topological}
Bergomi, M.G., Ferri, M., Zuffi, L.: Topological graph persistence.
\newblock Communications in Applied and Industrial Mathematics \textbf{11}(1),
  72--87 (2020).
\newblock \doi{doi:10.2478/caim-2020-0005}.
\newblock \urlprefix\url{https://doi.org/10.2478/caim-2020-0005}

\bibitem{bergomi2019rank}
Bergomi, M.G., Vertechi, P.: Rank-based persistence.
\newblock Theory and applications of categories \textbf{35}(9), 228--260 (2020)

\bibitem{blevins2020reorderability}
Blevins, A.S., Bassett, D.S.: Reorderability of node-filtered order complexes.
\newblock Physical Review E \textbf{101}(5), 052311 (2020)

\bibitem{BuSc14}
Bubenik, P., Scott, J.A.: Categorification of persistent homology.
\newblock Discrete \& Computational Geometry \textbf{51}(3), 600--627 (2014)

\bibitem{ChaCo*09}
Chazal, F., Cohen-Steiner, D., Glisse, M., Guibas, L.J., Oudot, S.Y.: Proximity
  of persistence modules and their diagrams.
\newblock In: SCG '09: Proceedings of the 25th annual symposium on
  Computational geometry, pp. 237--246. ACM, New York, NY, USA (2009).
\newblock \doi{http://doi.acm.org/10.1145/1542362.1542407}

\bibitem{chowdhury2018persistent}
Chowdhury, S., M{\'e}moli, F.: Persistent path homology of directed networks.
\newblock In: Proceedings of the Twenty-Ninth Annual ACM-SIAM Symposium on
  Discrete Algorithms, pp. 1152--1169. SIAM (2018)

\bibitem{CoEdHa07}
Cohen-Steiner, D., Edelsbrunner, H., Harer, J.: Stability of persistence
  diagrams.
\newblock Discr. Comput. Geom. \textbf{37}(1), 103--120 (2007).
\newblock \doi{http://dx.doi.org/10.1007/s00454-006-1276-5}

\bibitem{dAFrLa10}
d'Amico, M., Frosini, P., Landi, C.: Natural pseudo-distance and optimal
  matching between reduced size functions.
\newblock Acta {A}pplicandae {M}athematicae \textbf{109}(2), 527--554 (2010)

\bibitem{dereich2009random}
Dereich, S., M{\"o}rters, P.: Random networks with sublinear preferential
  attachment: degree evolutions.
\newblock Electronic Journal of Probability \textbf{14}, 1222--1267 (2009)

\bibitem{EdHa08}
Edelsbrunner, H., Harer, J.: Persistent homology---a survey.
\newblock In: Surveys on discrete and computational geometry, \emph{Contemp.
  Math.}, vol. 453, pp. 257--282. Amer. Math. Soc., Providence, RI (2008)

\bibitem{FrLaMe19}
Frosini, P., Landi, C., M{\'e}moli, F.: The persistent homotopy type distance.
\newblock Homology, Homotopy and Applications \textbf{21}(2), 231--259 (2019).
\newblock \doi{10.4310/HHA.2019.v21.n2.a13}.
\newblock \urlprefix\url{https://dx.doi.org/10.4310/HHA.2019.v21.n2.a13}

\bibitem{FrMu99}
Frosini, P., Mulazzani, M.: Size homotopy groups for computation of natural
  size distances.
\newblock Bull. of the Belg. Math. Soc. \textbf{6}(3), 455--464 (1999)

\bibitem{galashin2013existence}
Galashin, P.: Existence of a persistent hub in the convex preferential
  attachment model.
\newblock Probability and Mathematical Statistics \textbf{36}(1), 59--74 (2016)

\bibitem{galeana2014existence}
Galeana-S{\'a}nchez, H., Hern{\'a}ndez-Cruz, C.: On the existence of (k,
  l)-kernels in infinite digraphs: A survey.
\newblock Discussiones Mathematicae Graph Theory \textbf{34}(3), 431--466
  (2014)

\bibitem{govc2021complexes}
Govc, D., Levi, R., Smith, J.P.: Complexes of tournaments, directionality
  filtrations and persistent homology.
\newblock Journal of Applied and Computational Topology \textbf{5}(2), 313--337
  (2021)

\bibitem{kim2021generalized}
Kim, W., M{\'e}moli, F.: Generalized persistence diagrams for persistence
  modules over posets.
\newblock Journal of Applied and Computational Topology \textbf{5}(4), 533--581
  (2021)

\bibitem{Ku16}
Kurlin, V.: A fast persistence-based segmentation of noisy 2d clouds with
  provable guarantees.
\newblock Pattern Recognition Letters \textbf{83}, 3--12 (2016)

\bibitem{Les15}
Lesnick, M.: The theory of the interleaving distance on multidimensional
  persistence modules.
\newblock Foundations of Computational Mathematics pp. 1--38 (2015).
\newblock \doi{10.1007/s10208-015-9255-y}.
\newblock \urlprefix\url{http://dx.doi.org/10.1007/s10208-015-9255-y}

\bibitem{LoEx*16}
Lord, L.D., Expert, P., Fernandes, H.M., Petri, G., Van~Hartevelt, T.J.,
  Vaccarino, F., Deco, G., Turkheimer, F., Kringelbach, M.L.: Insights into
  brain architectures from the homological scaffolds of functional connectivity
  networks.
\newblock Frontiers in Systems Neuroscience \textbf{10} (2016)

\bibitem{mccleary2020bottleneck}
McCleary, A., Patel, A.: Bottleneck stability for generalized persistence
  diagrams.
\newblock Proceedings of the American Mathematical Society \textbf{148},
  3149--3161 (2020).
\newblock \doi{10.1090/proc/14929}

\bibitem{mccleary2020edit}
McCleary, A., Patel, A.: Edit distance and persistence diagrams over lattices.
\newblock SIAM Journal on Applied Algebra and Geometry \textbf{6}(2), 134--155
  (2022).
\newblock \doi{10.1137/20M1373700}.
\newblock \urlprefix\url{https://doi.org/10.1137/20M1373700}

\bibitem{morgenstern1953theory}
Morgenstern, O., Von~Neumann, J.: Theory of games and economic behavior.
\newblock Princeton university press (1953)

\bibitem{oudot2015persistence}
Oudot, S.Y.: Persistence theory: from quiver representations to data analysis,
  vol. 209.
\newblock American Mathematical Society Providence, RI (2015)

\bibitem{patel2018generalized}
Patel, A.: Generalized persistence diagrams.
\newblock Journal of Applied and Computational Topology \textbf{1}(3-4),
  397--419 (2018)

\bibitem{PeEx*14}
Petri, G., Expert, P., Turkheimer, F., Carhart-Harris, R., Nutt, D., Hellyer,
  P.J., Vaccarino, F.: Homological scaffolds of brain functional networks.
\newblock Journal of The Royal Society Interface \textbf{11}(101), 20140873
  (2014)

\bibitem{port2018persistent}
Port, A., Gheorghita, I., Guth, D., Clark, J.M., Liang, C., Dasu, S., Marcolli,
  M.: Persistent topology of syntax.
\newblock Mathematics in Computer Science \textbf{12}(1), 33--50 (2018)

\bibitem{ReNo*17}
Reimann, M.W., Nolte, M., Scolamiero, M., Turner, K., Perin, R., Chindemi, G.,
  D{\l}otko, P., Levi, R., Hess, K., Markram, H.: Cliques of neurons bound into
  cavities provide a missing link between structure and function.
\newblock Frontiers in Computational Neuroscience \textbf{11}, 48 (2017)

\bibitem{rieck2018clique}
Rieck, B., Fugacci, U., Lukasczyk, J., Leitte, H.: Clique community
  persistence: A topological visual analysis approach for complex networks.
\newblock IEEE Transactions on Visualization and Computer Graphics
  \textbf{24}(1), 822--831 (2018)

\bibitem{de2018theory}
de~Silva, V., Munch, E., Stefanou, A.: Theory of interleavings on categories
  with a flow.
\newblock Theory and Applications of Categories \textbf{33}(1), 583--607 (2018)

\bibitem{SiGi*18}
Sizemore, A.E., Giusti, C., Kahn, A., Vettel, J.M., Betzel, R.F., Bassett,
  D.S.: Cliques and cavities in the human connectome.
\newblock Journal of Computational Neuroscience \textbf{44}(1), 115--145
  (2018).
\newblock \doi{10.1007/s10827-017-0672-6}.
\newblock \urlprefix\url{https://doi.org/10.1007/s10827-017-0672-6}

\end{thebibliography}
\end{document}